\documentclass[aps,twocolumn]{revtex4-1}

\usepackage{amsmath,amssymb,amsfonts,amsthm,graphicx,float}
\usepackage{dcolumn}
\usepackage{bm}
\usepackage{subfigure}
\usepackage{braket}
\usepackage{xcolor}
\usepackage{tabularx}
\usepackage{multirow}
\usepackage[latin1]{inputenc}

\newcommand{\comment}[1]{{}}

\mathchardef\mhyphen="2D

\newcounter{thm}
\theoremstyle{definition}

\theoremstyle{plain}

\newtheorem{theo}[thm]{Theorem}

\begin{document}

\title{Certifying bipartite pure quantum states efficiently using untrusted devices}
 \author{Lijinzhi Lin$^{1,2}$}
 \author{Zhenyu Chen$^{2}$}
 \author{Xiaodie Lin$^{2}$}
 \author{Zhaohui Wei$^{3,4,}$}\email{weizhaohui@gmail.com}
 \affiliation{$^{1}$Department of Computer Science and Technology, Tsinghua University, Beijing 100084, China\\$^{2}$Institute for Interdisciplinary Information Sciences, Tsinghua University, Beijing 100084, China\\$^{3}$Yau Mathematical Sciences Center, Tsinghua University, Beijing 100084, China\\$^{4}$Yanqi Lake Beijing Institute of Mathematical Sciences and Applications, Beijing 101407, China}

\begin{abstract}
It has been known that all bipartite pure quantum states can be certified by quantum self-testing, i.e., any such states can be pinned down completely based on the statistics produced by local quantum measurements. A notable feature of quantum self-testing is that the conclusions remain reliable even when the quantum measurements involved are untrusted, where quantum nonlocality is crucial. This necessitates that each party conducts at least two different quantum measurements to produce the desired correlation. Here, we prove that when the underlying Hilbert space dimension is known beforehand, which is very common in quantum experiments, an arbitrary $d\times d$ bipartite pure state can be certified completely (up to local unitaries) by a certain correlation generated with a single measurement setting on each party, where each measurement yields only $2d$ or even $d+1$ outcomes. We also prove the robustness of our protocols to quantum noises and experimental imperfections. Compared with quantum self-testing, our protocols do not hinge on quantum nonlocality and are much more efficient, yet they maintain the essential feature of not requiring additional assumptions about the quantum devices involved. This advancement could offer significant convenience when certifying bipartite quantum states using untrusted quantum devices in future quantum industries.
    \end{abstract}

    \maketitle


\noindent In both quantum physics and quantum information processing, characterizing unknown quantum states is a fundamental task. A possible method for this task is quantum state tomography (QST), which aims to completely reconstruct the density matrices of target quantum states by performing tomographically complete quantum measurements and then analyzing the statistics~\cite{chuang1997prescription,poyatos1997complete}.
{However, in addition to the extremely high cost, QST protocols also require high precision regarding the involved quantum measurements, making their implementations very challenging~\cite{sugiyama2013precision}}.

Because of the limited precisions of devices and controls in modern quantum experiments, when certifying quantum states it will be desirable to minimize the reliance on the characterizations of the quantum operations involved. A class of protocols that operate on minimal physical assumptions on quantum devices are device-independent quantum protocols, which were first introduced in quantum key distribution~\cite{ekert1991quantum,barrett2005no,acin2006bell}.
In such protocols, properties of unknown quantum states can be inferred purely based on observed quantum nonlocality. An extreme example of this approach is quantum self-testing. In Ref.\cite{popescu1992states} (see also Ref.\cite{mayers2004self}), it was found that when a quantum state violates the Clauser-Horne-Shimony-Holt (CHSH) inequality maximally~\cite{clauser1969proposed}, i.e., the underlying nonlocality is the strongest, both the quantum state itself and the associated quantum measurements can be completely determined up to local isometries. Quantum self-testing enables reliable benchmarking of entangled quantum states even if the involved quantum measurements are untrustworthy. As a result, it is especially valuable in many quantum information processing tasks~\cite{ekert1991quantum,barrett2005no,acin2006bell}.

Since the discovery of quantum self-testing, a few other classes of bipartite and multipartite pure states have also been shown to be self-testable~\cite{yang2013robust,bamps2015sum,salavrakos2016bell,yang2014robust,acin2002quantum,mckague2014self,wu2014robust,pal2014device,vsupic2018self,baccari2020scalable}. Recently, a seminal work in Ref.\cite{coladangelo2017all} proved that an arbitrary $d\times d$ bipartite pure state can be self-tested, where each party conducts 3 or 4 measurements with $d$ outcomes each. Furthermore, in another breakthrough on quantum self-testing, by examining correlations generated by quantum networks, it has also been proved that all multipartite pure states can be self-tested with the help of ancillary parties and states~\cite{vsupic2023quantum}.
On the other hand, it has been known that mixed quantum states cannot be self-tested, as any correlation produced by a mixed state can always be {generated} by another pure state of the same size~\cite{sikora2016minimum}.

Quantum nonlocality is essential in quantum self-testing, which means in such protocols each party must perform quantum measurements in at least two different settings. Indeed, if each party has only one measurement setting, the resulting correlation can always be reproduced classically~\cite{zhang2012quantum}, rendering quantum self-testing impossible. However, due to various possible loopholes in experiments, exhibiting convincing quantum nonlocality is a very challenging task {\cite{hensen2015loophole,shalm2015strong,giustina2015significant}}. In such a situation, a profound question arises: if we cannot put too much trust in quantum devices, is quantum nonlocality necessary for quantum protocols certifying unknown quantum states solely based on observed correlations?

Here, we show that if the underlying Hilbert space dimension is known beforehand, which is very common in quantum experiments, any bipartite pure state can be certified completely (up to local unitaries) based only on a certain correlation produced by measuring the state with only one measurement setting on each party. In such protocols, quantum nonlocality cannot be produced, but a crucial feature similar to quantum self-testing remains, i.e., any other characterizations of quantum devices are not needed. In other words, in addition to quantum dimensions, the protocols we propose in this work can certify quantum states without any prior knowledge on target quantum states (i.e., they are pure or not) and the quality of involved quantum measurements. We call such a phenomenon \emph{semi-device-independent state certification}, or \emph{certification} for short, which is illustrated in FIG.~\ref{fig:scheme}. We stress that when certifying a $d\times d$ bipartite quantum pure state, the only local quantum measurement conducted by each party yields only $2d$ or even $d+1$ outcomes, which is far from being tomographically complete. We achieve this by exploiting a close relationship between the mathematical structure of correlations generated by local quantum measurements and that of positive semidefinite factorizations, a concept originated from optimization theory~\cite{fiorini2015exponential,gouveia2013lifts,jain2013efficient}.\\

\noindent {\large \textbf{Certifying the maximally entangled state}}

\begin{figure}[!t]
\centering
\includegraphics[width=0.35\textwidth]{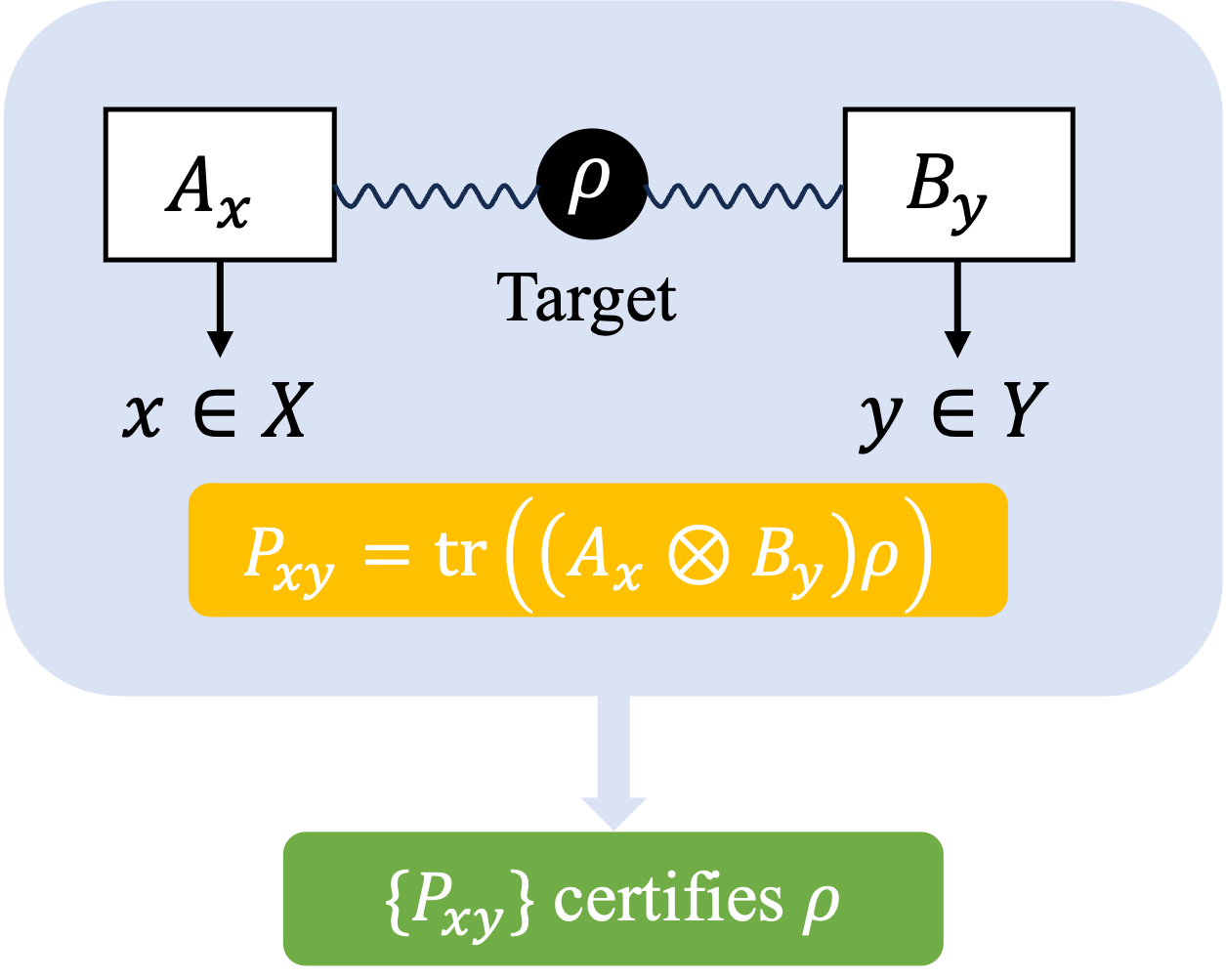}
\caption{The scheme of semi-device-independent state certification. The target quantum state $\rho$ with local dimension $d$ is shared by Alice and Bob, each of whom chooses a single setting to locally measure $\rho$. The obtained outcome statistics $\{P_{xy}\}$ can be utilized to certify $\rho$.}\label{fig:scheme}
\end{figure}

\noindent Suppose $\rho\in \mathcal{H}_A \otimes \mathcal{H}_B$ is a bipartite quantum state shared by Alice and Bob, and the local dimension is $d$. Each of them measures the corresponding subsystem with a local quantum measurement, and the probability that Alice and Bob obtain the outcome $x$ and $y$ respectively is
\begin{equation}\label{generation}
P_{xy}=\mathrm{tr}(\rho(A_x\otimes B_y)),
\end{equation}
where $A_x$ and $B_y$ are the measurement operators for Alice and Bob respectively, as illustrated in FIG.~\ref{fig:scheme}. We call the above matrix a \emph{correlation}, denoted $P$. For example, a specific correlation of size $2d\times 2d$ can be expressed as
\begin{align}\label{correlation:maximal}
    		P=
    		\frac{1}{d(d+1)^2}\begin{pmatrix}
    			d^2 I_d & \vec{e}\vec{e}^T \\
    			\vec{e}\vec{e}^T & I_d \\
    		\end{pmatrix},
\end{align}
where $\vec{e}$ is the $d$-dimensional all-one column vector and $I_d$ is the $d$-dimensional identity matrix. This correlation can be generated by locally measuring a $d\times d$ maximally entangled state $\sum_i \frac{1}{\sqrt{d}}\ket{ii}$, as long as  the measurement operators are chosen to be $A_i=B_i=\frac{d}{d+1}\ket{i-1}\bra{i-1}$, and $A_{d+i}=B_{d+i}=\frac{1}{d+1}\ket{\gamma_i}\bra{\gamma_i}$ for any $i\in[d]{\equiv\{1,2,...,d\}}$, where $\ket{\gamma_i}$ represents the $i$-th column of the $d\times d$ quantum Fourier transform matrix \cite{nielsen2010quantum}.

Then we have the following conclusion.
\begin{theo}\label{thm:maximal}
If Alice and Bob generate $P$ in Eq.\eqref{correlation:maximal} by locally measuring a $d\times d$ bipartite quantum state $\rho$, then $\rho$ must be a maximally entangled pure state.
\end{theo}


{We stress that Theorem \ref{thm:maximal} indicates a semi-device-independent state certification protocol, which means that in addition to the quantum dimension $d$, the certification works even if we do not make any other assumptions on the target state $\rho$ (i.e., it is pure or not), and we do not put any trust in the quantum measurements either.}

{We now sketch the proof for Theorem \ref{thm:maximal}, and more details can be seen in the Supplementary Materials. We first introduce several useful concepts. For an arbitrary $n\times n$ correlation $Q$, if there exist $r\times r$ positive semidefinite (PSD) matrices $C_1, C_2, \ldots, C_{n}, D_1, D_2, ..., D_{n}$ such that $Q_{xy}=\mathrm{tr}(C_x D_y)$ for any $x,y\in[n]$, we say that $\{C_x\}$ and $\{D_y\}$ form a PSD factorization for $Q$. If we further have $\sum_{x=1}^{n}C_x=\sum_{y=1}^{n}D_y=\Lambda=\text{diag}(a_1,a_2,...,a_r)$,
where $\text{diag}(a_1,a_2,...,a_r)$ is a diagonal matrix with the descending diagonal entries $a_i$'s, then we say that $\{C_x\}$ and $\{D_y\}$ form a canonical form of PSD factorization for $Q$. The minimum number $s$ such that $Q$ admits a PSD factorization of size $s\times s$ is called the PSD rank of $Q$.}

{Suppose $\ket{\psi}$ is a $t\times t$ bipartite pure quantum state, and its Schmidt coefficients are $\sqrt{\lambda_1}\geq\sqrt{\lambda_1}\geq...\geq\sqrt{\lambda_t}$. Then it has been known that $Q$ can be generated by locally measuring $\ket{\psi}$, if and only if $Q$ admits a canonical form of PSD factorization with $\Lambda=\text{diag}(\sqrt{\lambda_1},...,\sqrt{\lambda_t})$~\cite{chen2023generations}. Accordingly, if $Q$ can be generated by locally measuring a $t\times t$ mixed quantum state $\sigma$, and the spectrum of the
reduced state $\mathrm{tr}_A(\sigma)$ (or, by symmetry, $\mathrm{tr}_B(\sigma)$) is $a_i$($i\in[t]$), then by considering purifications of $\sigma$ we can see that $Q$ admits a canonical form of PSD factorization with $\Lambda=\text{diag}(\sqrt{a_1},...,\sqrt{a_t})$.}

{In addition, suppose $Q$ has a PSD rank $t$, we have the following fact on its canonical forms of PSD factorization. If $\{C_x\}$ and $\{D_y\}$ form a $t\times t$ PSD factorization, then for any $t\times t$ invertible matrix $H$, the operators $\{H C_x H^{\dagger}\}$ and $\{(H^{\dagger})^{-1} D_y H^{-1}\}$ also form a PSD factorization of $Q$. We call two PSD factorizations \textit{equivalent} if they are related as above by an invertible matrix. As shown in the Supplementary Materials, every PSD factorization $\{C_x\}$ and $\{D_y\}$ for $Q$ is equivalent to a canonical form of PSD factorization with the squared diagonal entries in the associated $\Lambda$ matching the eigenvalues of $\big(\sum_x C_x\big)\big(\sum_y D_y\big)$.}

{We now examine the correlation $P$ in Eq. (\ref{correlation:maximal}). Since $P$ can be generated by locally measuring a $d\times d$ quantum state, its PSD rank is at most $d$~\cite{jain2013efficient}; on the other hand, the upper left $d\times d$ submatrix of $P$ implies that the PSD rank of $P$ is at least $d$~{\cite{10.1007/s10107-015-0922-1}, hence the PSD rank of $P$ is exactly $d$. Then we can prove Theorem \ref{thm:maximal} as below. We first show that $P$ has a unique $d\times d$ canonical form of PSD factorization with $\Lambda=I_d/\sqrt{d}$. This step alone is not sufficient to certify $\rho$, as it still can be a mixed state whose reduced density matrices are $I_d/d$. Next, we prove that $\rho$ is pure, which implies that it must be a $d\times d$ maximally entangled pure state.}


{As mentioned above, in order to prove that $P$ has a unique canonical form of PSD factorization, we can show that for an arbitrary PSD factorization $\{C_x\}$ and $\{D_y\}$ for $P$, the eigenvalues of $\big(\sum_x C_x\big)\big(\sum_y D_y\big)$ are always the same. Indeed, by the lower right and the upper left submatrices in $P$,
one can see that all the corresponding $C_x$ and $D_y$ must be rank-1.} {Furthermore, {upon} adjusting $C_x$ and $D_y$ with a proper invertible matrix $H$,
i.e., $C_x'=HC_xH^\dag$ and $D_y'=(H^\dag)^{-1}D_yH^{-1}$,
one can obtain another PSD factorization for $P$,
which satisfies that $C_{d+x}'=D_{d+x}'=\frac{1}{\sqrt{d}(d+1)}\ket{x-1}\bra{x-1}$ for any $x\in[d]$. In addition, based on the constraints of generating the corresponding entries of $P$, it can be verified that under the adjustment of $H$ we also have $C_{x}'=D_{x}'$ for any $x\in[d]$, and they also have rank 1. As a result, it holds that $\sum_{x=1}^{d}C_x'=\sum_{y=1}^{d}D_y'=\frac{\sqrt{d}}{d+1}I_d$. Therefore, we eventually obtain that $\Lambda=\sum_{x=1}^{2d}C_x'=\sum_{y=1}^{2d}D_y'=I_d/\sqrt{d}$,
implying that the new PSD factorization is actually in a canonical form,
hence all PSD factorizations for $P$ of size $d\times d$ are equivalent to the same canonical form.}

To show that $\rho$ must be pure, we define $\tilde{B}_y=\mathrm{tr}_{B}((I_A\otimes {B_y})\rho)$, then Eq.\eqref{generation} implies that $P_{xy}=\mathrm{tr}(A_x\tilde{B}_y)$. Using similar techniques as before, we can prove that $A_x$ must be rank-1 for any $x\in[2d]$, and by symmetry, the same conclusion can be drawn for $B_y$.
{Then there exists local unitary transforms that sends both $\rho$ and all the measurement operators involved to a configuration where $A_i=B_i=\frac{d}{d+1}\ket{i-1}\bra{i-1}$ for any $i\in[d]$.}

{The zero entries of $P$ confines the support of $\rho$ inside} the space spanned by $\{\ket{00},\ket{11},\cdots,\ket{(d-1)(d-1)}\}$.
Lastly, by analyzing the entries of $A_{d+1}$ and $B_{d+1}$, we eventually obtain that
	\begin{align*}
		\braket{ii|\rho|i'i'}=\alpha\braket{i'|A_{d+1}|i}\braket{i'|B_{d+1}|i},
	\end{align*}
where $A_{d+1}$ and $B_{d+1}$ have rank $1$, and $\alpha$ is a scalar {independent of $i$ and $i'$}. This implies that $\rho$ must be pure, which concludes the proof for Theorem \ref{thm:maximal}.

We remark that, for a generic correlation, merely having a unique canonical form of PSD factorization is still far from being able to certify a quantum state. This is illustrated by the following simple example. Suppose the correlation $P=\frac{1}{2}I_2$, and the quantum dimension is 2. Then it can be proved that $P$ has a unique canonical form of PSD factorization. However, the correlation $P$ cannot certify any quantum state, as both $\frac{1}{2}(\ket{00}\bra{00}+\ket{11}\bra{11})$ and $\frac{1}{2}(\ket{00}+\ket{11})(\bra{00}+\bra{11})$ can generate $P$ when measured in the computational basis. \\

\noindent {\large \textbf{The rescaling technique for certifying general bipartite pure states}}

\noindent By developing a rescaling technique, we now prove that a general $d\times d$ bipartite pure state can be certified by correlations of size $2d\times 2d$ {in a semi-device-independent manner}. For convenience, we suppose its Schmidt rank is $d$, and the Schmidt coefficients are $\sqrt{\lambda_i}>0$, where $i\in[d]$.
\begin{theo}
		\label{thm:general2d}
		For any real number $\alpha_0\in(0,\frac{d+1}{\sqrt{d}}\min_{i}\sqrt{\lambda_i})$,
		let $\alpha_i= (d+1)\sqrt{d}\sqrt{\lambda_i}-\alpha_0\cdot d$ for any $i\in[d]$.
		Let $\vec{\alpha}=  (\alpha_0,\cdots,\alpha_0,\alpha_1,\cdots,\alpha_d)$.
		If the correlation $\tilde{P}=  \mathrm{diag}(\vec{\alpha})\cdot P\cdot\mathrm{diag}(\vec{\alpha})$
		is obtained upon locally measuring a $d\times d$ quantum state $\rho$, where $P$ is the correlation given in Eq.\eqref{correlation:maximal}, then $\rho$ must be a bipartite pure entangled state
		with the Schmidt coefficients $\sqrt{\lambda_1},\cdots,\sqrt{\lambda_d}$.
	\end{theo}

Note that the PSD rank of $\tilde{P}$ is $d$. To prove this result, we aim to demonstrate that $\tilde{P}$ has a unique canonical form of PSD factorization of size $d\times d$, where the diagonal entries of the corresponding $\Lambda$ are precisely $\sqrt{\lambda_i}$. This is achieved using the correlation rescaling technique. A key observation here is that when rescaling $P$ to obtain $\tilde{P}$, certain key properties of the PSD factorizations for $P$ can be directly carried over to the PSD factorizations of $\tilde{P}$. This transferability arises because applying proper scalar factors to each factorization operator for $\tilde{P}$ leads to valid factorization operators for $P$ (and vice versa, since any rescaling can always be reversed by another rescaling operation).

More specifically, suppose $C_x$ and $D_y$ form a PSD factorization for $\tilde{P}$, then $(1/\vec{\alpha}_x)C_x$ and $(1/\vec{\alpha}_y)D_y$ constitute a PSD factorization for $P$, where $\vec{\alpha}_x$ is the $x$-th entry of $\vec{\alpha}$. This allows us to characterize $C_x$ and $D_y$ based on the results in the previous section. Eventually, we can prove that $\tilde{P}$ has a unique $d\times d$ canonical form of PSD factorization, and the fact that $\rho$ is pure can be proved similarly. More details on the proof can be seen in the Supplementary Materials.\\

\noindent {\large \textbf{The robustness of our protocols}}

\noindent {Before applying our protocols in quantum experiments to certify quantum states, we have to make sure that they are robust to quantum noises and experimental imperfections.
{Since all involving operators during correlation generations are bounded and have finite dimension, our protocols are automatically robust, although this does not supply experimentalists with a concrete and useful robustness conclusion. For this, in the Supplemental Materials we prove the following theorem as an exemplary theoretical analysis for the robustness of our protocols}.
\begin{theo}\label{thm:robustness}
If Alice and Bob generate a $2d\times 2d$ correlation $P'$ by locally measuring a $d\times d$ bipartite quantum state $\rho$ such that $|P_{xy}-P'_{xy}|\leq\epsilon$ for any $x,y\in[2d]$, where $P$ is the correlation in Eq.\eqref{correlation:maximal} and $\epsilon$ is very small, then the fidelity between $\rho$ and a maximally entangled pure state will be at least $1-O(d^{4}\epsilon^{1/8})$.
\end{theo}}

{Furthermore, based on the above theorem, in the Supplemental Materials we show that the robustness of the protocol in Theorem \ref{thm:general2d} can also be proved. We would like to point out that stronger theoretical results on this important issue are expected. In the Supplemental Materials, we provide numerical evidence for this possibility, where a much stronger robustness can be observed.} \\

\noindent {\large \textbf{State certification with the minimum correlations}}

\noindent We would like to point out that the efficiency of {semi-device-independent} state certifications can be improved further. Actually, it turns out that the $d\times d$ maximally entangled pure state can be certified by certain correlations of size $(d+1)\times(d+1)$. Since any correlations of size $d\times d$ can always be generated by a $d\times d$ separable state~\cite{zhang2012quantum}, this is the best efficiency one can expect.

Consider the correlation $R$ of size $(d+1)\times(d+1)$ given by
	\begin{align}\label{correlation:maximal2}
		R(i,j)= &
		\begin{cases}
			{\frac{d}{(d+1)^2}}, & i=j, \\
			{\frac{1}{d(d+1)^2}}, & i\neq j.
		\end{cases}
	\end{align}
Then we have the following result.
\begin{theo}\label{thm:maximal2}
If Alice and Bob generate the correlation $R$ in Eq.\eqref{correlation:maximal2} by locally measuring a $d\times d$ bipartite quantum state $\rho$, then $\rho$ must be a maximally entangled pure state.
\end{theo}
The proof for this conclusion is quite involved and can be seen in the Supplemental Materials.
Moreover, we can also utilize the rescaling technique to certify other bipartite pure quantum states with correlations of size $(d+1)\times(d+1)$. In fact, we have proved that general bipartite pure states can be certified with $(d+1)\times(d+1)$ correlations in an approximate manner, i.e., the set of bipartite pure quantum states that can be certified via such correlations of form $\mathrm{diag}(\vec{\alpha})R\mathrm{diag}(\vec{\beta})$ is dense, where $\vec{\alpha}$ and $\vec{\beta}$ are vectors with all entries positive. We conjecture that exact certification for these quantum states is also possible.\\

\noindent {\large \textbf{Discussion}}

\noindent We have proven that if the dimension of a bipartite quantum system is known, any pure state of this system can be certified based on the statistics produced by local quantum measurements, where each party involves only one single measurement setting. As a result, our protocols do not rely on quantum nonlocality, and this may bring remarkable convenience to their physical implementations. Meanwhile, these protocols retain the crucial feature that they do not rely on other assumptions about involved quantum devices except quantum dimensions. Particularly, we do not need any other prior knowledge on target states (i.e., they are pure or not), and we do not need to trust the quality of quantum measurements either. We stress that in order to certify a $d\times d$ bipartite pure quantum state, each quantum measurement in our protocols generates only $2d$ or even $d+1$ outcomes. Compared with quantum self-testing, the new protocols provide us a more efficient approach to certify bipartite pure quantum states using untrusted quantum devices.


{To apply our protocols as a realistic and practical approach that certifies crucial quantum states, an important issue we have to address is how robust to noise these protocols can be. We have proved that if the correlation observed in a quantum experiment is not exactly the correlation $P$ we discussed but very close, we can obtain a lower bound for the fidelity of the target quantum state. Lastly, we anticipate the development of similar protocols for certifying multipartite pure quantum states. We leave this problem for future work.}

\begin{acknowledgements}
{We thank Zhengfeng Ji, Yongjian Han, Zizhu Wang, Yu Guo, and Zhangqi Yin for helpful comments on our manuscript.} This work was supported by the National Natural Science Foundation of China, Grant Nos. 62272259, 61832015, 62332009, and the National Key R\&D Program of China, Grants Nos. 2018YFA0306703, 2021YFE0113100.\\
\end{acknowledgements}

\bibliographystyle{plain}
\bibliography{ref}

\end{document}


\title{Supplemental Materials: Certifying bipartite pure quantum states efficiently using untrusted devices}

\author{Lijinzhi Lin$^{1,2}$}
\author{Zhenyu Chen$^{2}$}
\author{Xiaodie Lin$^{2}$}
\author{Zhaohui Wei$^{3,4,}$}\email{weizhaohui@gmail.com}
\affiliation{$^{1}$Department of Computer Science and Technology, Tsinghua University, Beijing 100084, China\\$^{2}$Institute for Interdisciplinary Information Sciences, Tsinghua University, Beijing 100084, China\\$^{3}$Yau Mathematical Sciences Center, Tsinghua University, Beijing 100084, China\\$^{4}$Yanqi Lake Beijing Institute of Mathematical Sciences and Applications, Beijing 101407, China}

\begin{abstract}
	In the following, we provide the proofs for the theorems in the main text. {Particularly, we will prove the robustness of our protocols, and elaborate on the details of the correlation rescaling technique we develop for certifying general bipartite pure quantum states.}
\end{abstract}

\maketitle
\setcounter{figure}{0}
\renewcommand{\thefigure}{S\arabic{figure}}
\appendix
\section{Certifying the maximally entangled pure states with $2d\times 2d$ correlations}
\label{appendix:a}

    We define
    \begin{align}\label{correlation:maximal}
    		P=
    		\frac{1}{d(d+1)^2}\begin{pmatrix}
    			d^2 I_d & \vec{e}\vec{e}^T \\
    			\vec{e}\vec{e}^T & I_d \\
    		\end{pmatrix},
	\end{align}
	where $\vec{e}$ is the $d$-dimensional all-one column vector and $I_d$ is the identity matrix. If the smallest PSD factorization for $P$ is of size $r\times r$, we say that its PSD rank is $r$, and the lower right submatrix $\frac{1}{d(d+1)^2}I_d$ of $P$ suggests that $r\geq d$.
    \begin{theo}\label{thm:maximal}
    If Alice and Bob generate $P$ in Eq.\eqref{correlation:maximal} by locally measuring a $d\times d$ bipartite quantum state $\rho$, then $\rho$ must be a maximally entangled pure state.
    \end{theo}


    We first introduce several useful concepts. For an arbitrary $n\times n$ correlation $Q$, if there exist $r\times r$ positive semidefinite (PSD) matrices $C_1, C_2, \ldots, C_{n}, D_1, D_2, ..., D_{n}$ such that $Q_{xy}=\mathrm{tr}(C_x D_y)$ for any $x,y\in[n]$, we say that $\{C_x\}$ and $\{D_y\}$ form a PSD factorization for $Q$. If we further have $\sum_{x=1}^{n}C_x=\sum_{y=1}^{n}D_y=\Lambda=\text{diag}(a_1,a_2,...,a_r)$, where $\text{diag}(a_1,a_2,...,a_r)$ is a diagonal matrix with the descending diagonal entries $a_i$'s, then we say that $\{C_x\}$ and $\{D_y\}$ form a canonical form of PSD factorization for $Q$. The minimum number $s$ such that $Q$ admits a PSD factorization of size $s\times s$ is called the PSD rank of $Q$.

    {Suppose $\ket{\psi}$ is a $t\times t$ bipartite pure quantum state, and its Schmidt coefficients are $\sqrt{\lambda_1}\geq\sqrt{\lambda_1}\geq...\geq\sqrt{\lambda_t}$. Then it has been known that $Q$ can be generated by locally measuring $\ket{\psi}$, if and only if $Q$ admits a canonical form of PSD factorization with $\Lambda=\text{diag}(\sqrt{\lambda_1},...,\sqrt{\lambda_t})$~\cite{chen2023generations}. Accordingly, if $Q$ can be generated by locally measuring a $t\times t$ mixed quantum state $\sigma$, and the spectrum of the reduced state $\mathrm{tr}_A(\sigma)$ (or, by symmetry, $\mathrm{tr}_B(\sigma)$) is $a_i$($i\in[t]$), then by considering purifications of $\sigma$ we can see that $Q$ admits a canonical form of PSD factorization with $\Lambda=\text{diag}(\sqrt{a_1},...,\sqrt{a_t})$.}


    Then one can prove the above theorem by first showing that the canonical form of PSD
    factorization for $P$ of size $d\times d$ is unique with $\Lambda=I_d/\sqrt{d}$, and then proving that $\rho$ must be pure. For this, we need to examine all possible PSD factorizations for $P$ of size $d\times d$. It turns out the following lemma will be very useful for us. 

	\begin{lemm}
		\label{correlation:scsd}
		Let $Q$ be a correlation indexed by $x$ and $y$.
		Let $r$ be the PSD rank of $Q$.
		Let $\{C_x\}$, $\{D_y\}$ and $\{C'_x\}$, $\{D'_y\}$ be two $r\times r$ PSD factorizations of $Q$.
		We say that the above two PSD factorizations are equivalent if there exists an $r\times r$ invertible matrix $H$ such that
		\begin{align*}
			H C_x H^{\dagger}= & C'_x, \\
			(H^{\dagger})^{-1} D_y H^{-1}= & D'_y.
		\end{align*}
		Then:
		\begin{enumerate}
			\item
				If the equivalence holds, then the characteristic polynomial (hence also spectrum) of $\big(\sum_x C_x\big)\big(\sum_y D_y\big)$
				is equal to that of $\big(\sum_x C'_x\big)\big(\sum_y D'_y\big)$.
			\item
				Any PSD factorization $\{C_x\}$, $\{D_y\}$ of $Q$ is equivalent to a canonical form of PSD factorization $\{C''_x\}, \{D''_y\}$
				which satisfies
				\begin{align}
					\sum_{x} C''_x= & \sum_{y} D''_y= \mathrm{diag}(\sqrt{\lambda_1},\cdots,\sqrt{\lambda_r}),
				\end{align}
				where $\lambda_i$ are the eigenvalues of $\big(\sum_x C_x\big)\big(\sum_y D_y\big)$.
				{For convenience the matrix $\mathrm{diag}(\sqrt{\lambda_1},\cdots,\sqrt{\lambda_r})$ is henceforth denoted as $\Lambda$.}
		\end{enumerate}
	\end{lemm}
	
	\begin{proof}
    	\begin{enumerate}
			\item
            Follows directly from the definition of the equivalence relation and the cyclic property of characteristic polynomial.
			\item
				Let $S_C=\sum_x C_x$ and $S_D=\sum_y D_y$.
				We first show that $S_C$ and $S_D$ both have full rank.
				{
				Without loss of generality, suppose $S_C$ does not have full rank.
				This means that there exists an isometry $V:\mathbb{C}^{r-1}\to\mathbb{C}^{r}$ and a PSD operator $S'_C\in\mathbb{C}^{(r-1)\times(r-1)}$ such that $V S'_C V^{\dagger}=S_C$.
				For any $x$ there also exists a PSD operator $C'_x\in\mathbb{C}^{(r-1)\times(r-1)}$ such that $VC'_x V^{\dagger}=C_x$,
				since the image of $C_x$ is contained in that of $S_C$.
				Then for any $x$ and $y$ we have
				\begin{align*}
					\mathrm{tr}(C_x D_y)=\mathrm{tr}(VC'_xV^{\dagger} D_y)=\mathrm{tr}(C'_x (V^{\dagger}D_y V)),
				\end{align*}
				which implies the existence of a $(r-1)$-sized PSD decomposition for $Q$, contradiction.
				Thus $S_C$ and $S_D$ are both invertible.
				}

				Now consider
				\begin{align*}
					\hat{C}''_x= & \left(S_C^{1/2} S_D S_C^{1/2}\right)^{1/4}\left(S_C^{-1/2} C_x S_C^{-1/2}\right)\left(S_C^{1/2}S_D S_C^{1/2}\right)^{1/4}, \\
					\hat{D}''_y= & \left(S_C^{1/2} S_D S_C^{1/2}\right)^{-1/4}\left(S_C^{1/2} D_y S_C^{1/2}\right)\left(S_C^{1/2}S_D S_C^{1/2}\right)^{-1/4}.
				\end{align*}
				Then $\{\hat{C}''_x\}$ and $\{\hat{D}''_y\}$ form a PSD factorization of $Q$
				that is equivalent to $\{C_x\}$ and $\{D_y\}$;
				furthermore, we have $\sum_x \hat{C}''_x=\sum_y \hat{D}''_y=\left(S_C^{1/2} S_D S_C^{1/2}\right)^{1/2}$, implying that $(\sum_x \hat{C}''_x)^2$ has the same spectrum with $S_C S_D$, where we have utilized the fact that $S_C^{1/2} S_D S_C^{1/2}$ has the same spectrum with $S_C S_D$.
				Let $U$ be the unitary operator that diagonalizes $\sum_x \hat{C}''_x$, i.e., $U\left(\sum_x\hat{C}''_x\right)U^{\dagger}$ is a diagonal matrix.
				Then $\{U\hat{C}''_xU^{\dagger}\}$ and $\{U\hat{D}''_yU^{\dagger}\}$
				is the desired canonical form of PSD factorization for $Q$.
		\end{enumerate}
		
	\end{proof}
	
   The above lemma implies that for $P$ in Eq.\eqref{correlation:maximal}, if all the PSD factorizations of size $d\times d$ are equivalent to a same canonical form of PSD factorization, then {$P$ must have a unique canonical forms of PSD factorization of size $d\times d$}. 
   Indeed, two canonical forms of PSD factorization with different spectra are not equivalent to each other.\\
	
\noindent \textbf{The proof for Theorem 1 }We first consider the lower right submatrix $\frac{1}{d(d+1)^2}I_d$ in $P$. It can be verified that this implies that
    for $x,y\in[d]=\{1,2,...,d\}$, $C_x$ and $D_y$ must be rank-1, i.e., $C_{d+x}={\mathrm{tr}(C_{d+x})}\ket{\alpha_x}\bra{\alpha_x}$ and $D_{d+y}={\mathrm{tr}(D_{d+y})}\ket{\beta_y}\bra{\beta_y}$ for certain pure states $\ket{\alpha_x}$ and $\ket{\beta_y}$. Furthermore, the vector set $\{\ket{\alpha_x}\}(x\in[d])$ is linearly independent, and similar for $\{\ket{\beta_y}\}(y\in[d])$. To see why this is the case, for each $x\in[d]$ we pick unit vectors $a_x\in\mathrm{range}(C_{d+x})$ and $b_{{x}} \in\mathrm{range}(D_{d+{x}})$ such that $\langle a_x,b_x\rangle\neq 0$;
	this is possible since $\mathrm{tr}(C_{d+x}D_{d+x})>0$, {where $\mathrm{range}(A)$ is the subspace consisting of all linear combinations of the eigenvectors of $A$ with nonzero eigenvalues}. Suppose the vectors $b_x$ are linearly dependent; that is,
    without loss of generality,
	\begin{align*}
		b_1=\sum_{x=2}^{d}{c_x} b_x,
	\end{align*}
	where {$c_x$} are complex scalars. Then
	\begin{align*}
		0\neq \langle a_1,b_1\rangle=\sum_{x=2}^{d}{c_x}\langle a_1,b_x\rangle=0,
	\end{align*}
	which is impossible.
	This gives
	\begin{align*}
		\mathrm{rank}(D_{d+2}+\cdots+D_{2d})\geq\mathrm{dim}(\mathrm{span}\{b_2,\cdots,b_d\})=d-1,
	\end{align*}
	{where $\mathrm{rank}(A)$ is the rank of the matrix $A$, $\mathrm{dim}(S)$ is the dimension of a vector space $S$, and $\mathrm{span}\{b_2,\cdots,b_d\}$ is the vector space consisting of all linear combinations of the vectors $b_2,\cdots,b_d$.} Since $\mathrm{tr}(C_{d+1}(D_{d+2}+\cdots+D_{2d}))=0$, we have $\mathrm{rank}(C_{d+1})\leq 1$,
	hence $\mathrm{rank}(C_{d+1})=1$ due to $C_{d+1}\neq 0$.
	By symmetry we have $\mathrm{rank}(C_{d+x})=1$ for all $x\in[d]$,
	and by swapping the roles of $C_{d+x}$ and $D_{d+y}$ we get $\mathrm{rank}(D_{d+y})=1$ for all $y\in[d]$.

    {Let $H=\left(\sum_{x=1}^{d}\sqrt{\sqrt{d}(d+1)\mathrm{tr}(C_{d+x})}\ket{\alpha_x}\bra{x-1}\right)^{-1}$}. Then $HC_{d+x}H^\dag=\frac{1}{\sqrt{d}(d+1)}\ket{x-1}\bra{x-1}$, where $x\in[d]$ and {$\{\ket{i}\}(i\in\{0,1,...,d-1\})$ are the computational basis states}. For any $x,y\in[2d]$, we let $C_x'=HC_xH^\dag$ and {$D_y'=(H^\dag)^{-1}D_yH^{-1}$}, then $\{C_x'\}$ and $\{D_y'\}$ {is a PSD factorization for $P$ that is equivalent to $\{C_x\}$ and $\{D_y\}$}. {By verifying the entries of $P$ we obtain that} $D_{d+y}'=\frac{1}{\sqrt{d}(d+1)}\ket{y-1}\bra{y-1}$ for $y\in[d]$.

    Similarly, for any $x,y\in[d]$ the matrices $C_x'$ and $D_y'$ must be of rank $1$ since they have to produce the submatrix $\frac{d^2}{d(d+1)^2}I_d$ of $P$. Additionally, for any $x,y\in[d]$ the diagonals of $C_x'$ and $D_y'$ are all $\frac{1}{\sqrt{d}(d+1)}$ (from the two submatrices $\frac{1}{d(d+1)^2}\vec{e}\vec{e}^T$ of $P$), hence {$\mathrm{tr}(C_x')=\mathrm{tr}(D_y')=\frac{\sqrt{d}}{d+1}$ and} all entries of $C_x'$ and $D_y'$ have norms at most $\frac{1}{\sqrt{d}(d+1)}$. For any $x\in[d]$, since $\frac{d}{(d+1)^2}=\mathrm{tr}(C_x'D_x')\leq \sqrt{\mathrm{tr}(C_x'^2)\mathrm{tr}(D_x'^2)}\leq\mathrm{tr}(C_x')\mathrm{tr}(D_x')=\frac{d}{(d+1)^2}$
	with the equality holding if and only if $C_x'$ is proportional to $D_x'$, we have $C_x'=D_x'$. Thus the constraint
    $\mathrm{tr}(C_x'D_y')=\frac{d}{(d+1)^2}\delta_{xy}$ {from the correlation} implies that $\sum_{x=1}^dC_x'=\sum_{y=1}^dD_y'=\frac{\sqrt{d}}{d+1}I_d$.

    In fact, for any $x,y\in[d]$ if setting $C_x'=D_x'=\frac{\sqrt{d}}{d+1}\ket{\gamma_x}\bra{\gamma_x}$ where $\ket{\gamma_x}$ is the $x$-th column of the $d\times d$ quantum Fourier transform matrix, and then combining the above choices of $\{C_{d+x}'\}$ and $\{D_{d+y}'\}$ together, one can verify that $\{C_x'\}$ and $\{D_y'\}$ do form a valid PSD factorization for $P$. Furthermore, we also have that $\sum_{x=1}^{2d}C_x'=\sum_{y=1}^{2d}D_y'=\frac{1}{\sqrt{d}}I_d$, implying that $\{C_x'\}$ and $\{D_y'\}$ form a canonical form of PSD factorization for $P$. The following lemma summarizes the above observations.
    \begin{lemm}
		\label{lemm:unique}
        The spectrum of all the canonical forms of PSD factorization for $P$ of size $d\times d$ is unique.
	\end{lemm}

    Note that $\rho$ can be a mixed state, and if this is the case, Bob can purify $\rho$ to be a pure state $\ket{\phi}$ in $\h_A\otimes \h_B \otimes \h_{R}$ by introducing an ancillary system $R$. After that, Alice and Bob can generate $P$ by measuring {$\ket{\phi}$} with the POVMs $\{A_x\}$ and $\{B_y\otimes I_R\}$, where $I_R$ is the identity operator on the system $R$. According to Ref.[31], we know that with respect to the partition {$A|BR$}, the Schmidt coefficients of $\ket{\phi}$ is the diagonal entries of $\frac{1}{\sqrt{d}}I_d$, i.e., $\{\frac{1}{\sqrt{d}},\frac{1}{\sqrt{d}},...,\frac{1}{\sqrt{d}}\}$. Since the ancillary system $R$ can also be introduced by Alice, we obtain the following conclusion.
    \begin{prop}
		$\mathrm{tr}_A(\rho)=\mathrm{tr}_B(\rho)=\frac{1}{d}I_d$.
	\end{prop}

	Due to the above lemma, if we can next show that the state $\rho$ is pure, then it must be a bipartite maximally entangled pure state. For this, we need to use the following fact. {Recall that Alice and Bob generate $P$ by measuring $\rho$ with two local quantum measurements with operators $\{A_x\}$ and $\{B_y\}$ respectively.}
	\begin{prop}
		\label{lemm:rankone}
		All measurement operators $A_x$ and $B_y$ have rank $1$, where $x,y\in[2d]$.
	\end{prop}
	\begin{proof}
		Let $\tilde{B}_y=\mathrm{tr}_{B}((I_A\otimes B_y)\rho)$.
		Then $P_{xy}=\mathrm{tr}(A_x\tilde{B}_y)=\mathrm{tr}(\rho(A_x\otimes B_y))$, {hence $\{A_x\}$ and $\{\tilde{B}_y\}$ is a PSD decomposition for $P$.}
		Note that $\tilde{B}_y$ is a PSD matrix for any $y\in[2d]$, thus we can apply the technique we have utilized {in Theorem~\ref{thm:maximal}} to prove that $\mathrm{rank}(A_x)=1$ for any $x\in[2d]$. By swapping the roles of $A_x$ and $B_y$ we get $\mathrm{rank}(B_y)=1$ for all $y\in[2d]$.
	\end{proof}

    Then $A_x$ and $B_y$ can be further characterized as below.
	\begin{prop}
		For $i,j\in[d]$ with $i\neq j$, we have
		\begin{align}
			A_iA_j=B_iB_j=A_{d+i}A_{d+j}=B_{d+i}B_{d+j}=0.
		\end{align}
	\end{prop}
	\begin{proof}
        {
		We now show that $A_i A_j=0$ for $i,j\in[d]$ with $i\neq j$, and the other cases can be proved similarly.
        Recall that $P_{xy}=\mathrm{tr}(\rho(A_x\otimes B_y))$, then it holds that
        \begin{equation*}
        \sum_{y=1}^{2d}P_{iy}=\mathrm{tr}(\rho(A_i\otimes I_d))=\mathrm{tr}(\mathrm{tr}_B(\rho)\cdot A_i)=\frac{\mathrm{tr}(A_i)}{d}.
        \end{equation*}
        Combined with $\sum_{y=1}^{2d}P_{iy}=\frac{1}{d+1}$, this means that $\mathrm{tr}(A_i)=\frac{d}{d+1}$.
        Similarly $\mathrm{tr}(A_j)=\frac{d}{d+1}$.
        }

		{
		Let $\tilde{B}_y=\mathrm{tr}_{B}((I_A\otimes B_y)\rho))$ for $y\in[2d]$.
		We have
		\begin{align*}
			\mathrm{tr}(\tilde{B}_j)= & \sum_{x=1}^{2d}\mathrm{tr}(A_x \tilde{B}_j) \\
			= & \sum_{x=1}^{2d}\mathrm{tr}(A_x\mathrm{tr}_B((I_A\otimes B_j)\rho)) \\
			= & \sum_{x=1}^{2d}\mathrm{tr}((A_x\otimes B_j)\rho) \\
			= & \sum_{x=1}^{2d}P_{xj} \\
			= & \frac{1}{d+1}.
		\end{align*}
		Then
		\begin{align*}
			\frac{d}{(d+1)^2}=P_{jj}=\mathrm{tr}(A_j\tilde{B}_j)\leq \mathrm{tr}(A_j)\mathrm{tr}(\tilde{B}_j)=\frac{d}{(d+1)^2},
		\end{align*}
		which is possible only when $\tilde{B}_j$ is a scalar multiple of $A_j$.
		By $P_{ij}=\mathrm{tr}(A_i\tilde{B_j})=0$ we have $A_i\tilde{B_j}=0$, hence $A_i A_j=0$.
		}
	\end{proof}

	Because of the above fact, we may perform two local unitaries on each party so that
	for $i\in[d]$ it holds that $A_i$ and $B_i$ are both proportional to {$\ket{i-1}\bra{i-1}$},
	where $\rho$ is also adjusted accordingly.
	{
	Together with $\mathrm{tr}(A_i)=\mathrm{tr}(B_i)=\frac{d}{d+1}$ for $i\in[d]$, we can let {$A_i=\frac{d}{d+1}\ket{i-1}\bra{i-1}$ and $B_i=\frac{d}{d+1}\ket{i-1}\bra{i-1}$} for $i\in[d]$.
	}

	Lastly, since $P_{ij}=\mathrm{tr}(\rho(A_i\otimes B_j))=0$ for any $i\neq j\in[d]$, we have that $\mathrm{range}(\rho)\subseteq\mathrm{span}\{\ket{00},\ket{11},\cdots,\ket{(d-1)(d-1)}\}$. Similarly, it holds that $\braket{ii|\rho|ii}=\frac{1}{d}$, {which also means $|\braket{ii|\rho|jj}|\leq \frac{1}{d}$ since $\rho$ is a PSD matrix, here $i,j\in\{0,\cdots,d-1\}$.}

	\begin{prop}
		$|\braket{i|A_{d+1}|j}|=|\braket{i|B_{d+1}|j}|=\frac{1}{d(d+1)}$ for all {$i,j\in\{0,\cdots,d-1\}$}.
	\end{prop}
	\begin{proof}
		For any {$i\in\{0,\cdots,d-1\}$} we have
		\begin{align*}
			\frac{1}{d(d+1)^2}= & \mathrm{tr}(\rho({A_{i+1}}\otimes B_{d+1})) \\
			= & \frac{d}{d+1}\mathrm{tr}(\rho(\ket{i}\bra{i}\otimes B_{d+1})) \\
			= & \frac{d}{d+1}\sum_{jkj'k'}\braket{jk|\rho|j'k'}\bra{j'k'}(\ket{i}\bra{i}\otimes B_{d+1})\ket{jk} \\
			= & \frac{d}{d+1}\sum_{kk'}\braket{ik|\rho|ik'}\braket{k'|B_{d+1}|k} \\
			= & \frac{d}{d+1}\cdot\frac{1}{d}\braket{i|B_{d+1}|i} \\
			= & \frac{1}{d+1}\cdot\braket{i|B_{d+1}|i},
		\end{align*}
		hence $\braket{i|B_{d+1}|i}=\frac{1}{d(d+1)}$.
		Since $\mathrm{rank}(B_{d+1})=1$, the rest of the entries also have norms $\frac{1}{d(d+1)}$.
	\end{proof}

	We now examine the $(d+1,d+1)$-th entry of $P$ using the Cauchy-Schwarz inequality:
	\begin{align*}
		\frac{1}{d(d+1)^2}= & \mathrm{tr}(\rho(A_{d+1}\otimes B_{d+1})) \\
		= & \sum_{iji'j'}\braket{ij|\rho|i'j'}\braket{i'j'|A_{d+1}\otimes B_{d+1}|ij} \\
		= & \sum_{ii'}\braket{ii|\rho|i'i'}\braket{i'|A_{d+1}|i}\braket{i'|B_{d+1}|i} \\
		\leq & \left(\sum_{ii'}|\braket{ii|\rho|i'i'}|^2\right)^{1/2}\left(\sum_{ii'}|\braket{i'|A_{d+1}|i}\braket{i'|B_{d+1}|i}|^2\right)^{1/2} \\
		\leq & \left(d^2\cdot d^{-2}\right)^{1/2}\left(d^2\cdot (d(d+1))^{-4}\right)^{1/2} \\
		= & \frac{1}{d(d+1)^2}.
	\end{align*}
	Therefore, there exists a scalar $\alpha$ such that
	\begin{align}
		\braket{ii|\rho|i'i'}=\alpha\braket{i'|A_{d+1}|i}\braket{i'|B_{d+1}|i}.
	\end{align}
	{Recall that $A_{d+1}$ and $B_{d+1}$ have rank 1, and $\mathrm{range}(\rho)\subseteq\mathrm{span}\{\ket{00},\ket{11},\cdots,\ket{(d-1)(d-1)}\}$, then the above relation} implies that $\mathrm{rank}(\rho)=1$, hence $\rho$ is a pure state. By $\mathrm{tr}_A(\rho)=\mathrm{tr}_B(\rho)=\frac{I_d}{d}$, we conclude that $\rho$ is the maximally entangled state.

\section{Certifying general bipartite entangled pure states with $2d\times 2d$ correlations}

    Based on the result in the previous section, we now develop a rescaling technique to prove that general bipartite entangled pure states can be certified by $2d\times 2d$ correlations {in a semi-device-independent manner}.
	\begin{theo}
		\label{correlation:general2d}
		Suppose $\lambda_i>0$ for $i\in[d]$.
		For any real number $\alpha_0\in(0,\frac{d+1}{\sqrt{d}}\min_{i}\sqrt{\lambda_i})$,
		let
		\begin{align}
			\alpha_i= & (d+1)\sqrt{d}\sqrt{\lambda_i}-\alpha_0 d
		\end{align}
		for $i\in[d]$.
		Let
		\begin{align*}
			\vec{\alpha}= & (\alpha_0,\cdots,\alpha_0,\alpha_1,\cdots,\alpha_d).
		\end{align*}
		If the correlation
		\begin{align}
			\tilde{P}= & \mathrm{diag}(\vec{\alpha})P\mathrm{diag}(\vec{\alpha})
		\end{align}
		is obtained upon locally measuring a $d\times d$ state $\rho$, where $P$ is the correlation given in Eq.\eqref{correlation:maximal}, then $\rho$ is bipartite pure entangled state
		with Schmidt coefficients $\sqrt{\lambda_1},\cdots,\sqrt{\lambda_d}$.
	\end{theo}

	\begin{proof}
		By straightforward verification, one can show that $\tilde{P}$ is a valid correlation of size $2d\times 2d$, {i.e., all the entries of $\tilde{P}$ are nonnegative and they sum to 1.}
        Notice that $\tilde{P}$ has the same PSD rank of $d$ as $P$.
		In fact, let $\{C_i\}$, $\{D_j\}$ be {an arbitrary} $d\times d$ PSD factorization for $\tilde{P}$.
		Then $\{\vec{\alpha}_i^{-1}C_i\}$ and $\{\vec{\alpha}_j^{-1}D_j\}$ form a $d\times d$ PSD factorization for $P$,
		where $\vec{\alpha}_i$($i\in[2d]$) denotes the $i$-th component of $\vec{\alpha}$.
		{According to our discussions in the previous section, this} means that the PSD factorization $\{\vec{\alpha}_i^{-1}C_i\}$, $\{\vec{\alpha}_j^{-1}D_j\}$
		is equivalent to a PSD factorization $\{C'_i\}$, $\{D'_j\}$ of $P$
		such that
		\begin{align*}
			C'_i= & D'_i, & \forall i\in[2d], \\
			\sum_{i=1}^{d}C'_i= & \frac{\sqrt{d}}{d+1}I_d, \\
			C'_{d+i}= & \frac{1}{\sqrt{d}(d+1)}\ket{i-1}\bra{i-1}, & \forall i\in[d].
		\end{align*}
		The same equivalence transform sends $\{C_i\}$, $\{D_j\}$
		to $\{\vec{\alpha}_i C'_i\}$, $\{\vec{\alpha}_j D'_j\}$.
		{It can be verified that $\{\vec{\alpha}_i C'_i\}$, $\{\vec{\alpha}_j D'_j\}$ is actually a canonical form of PSD factorization for $\tilde{P}$ with the same spectrum as} 
		\begin{align}
			{\left(\sum_{i=1}^{2d}\vec{\alpha}_i C'_i\right)\left(\sum_{i=1}^{2d}\vec{\alpha}_j D'_j\right)=}\left(\alpha_0\frac{\sqrt{d}}{d+1}I_d+\frac{1}{(d+1)\sqrt{d}}\sum_{i=1}^{d}\alpha_i\ket{i-1}\bra{i-1}\right)^2,
		\end{align}
		{and the spectrum} by the constructions of $\alpha_0$ and $\alpha_i$($i\in[d]$) is exactly $\lambda_1,\cdots,\lambda_d$.
		Therefore, any bipartite $d\times d$ state $\rho$ that generates $\tilde{P}$
		must have its reduced states $\mathrm{tr}_A(\rho)$ and $\mathrm{tr}_B(\rho)$
		having the spectrum $\lambda_1,\cdots,\lambda_d$.

		Now let $A_i$ and $B_j$ be the local measurement operators that generate $\tilde{P}$ upon measuring $\rho$;
		that is,
		\begin{align*}
			\mathrm{tr}((A_i\otimes B_j)\rho)=\tilde{P}_{ij}.
		\end{align*}
		Then
		\begin{align*}
			\mathrm{tr}((\vec{\alpha}_i^{-1} A_i\otimes \vec{\alpha}_j^{-1} B_j)\rho)={P_{ij}}.
		\end{align*}
		Let $S_A=\sum_{i=1}^{2d}\vec{\alpha}_i^{-1}A_i$ and $S_B=\sum_{j=1}^{2d}\vec{\alpha}_j^{-1}B_j$. {Since $\vec{\alpha}_i>0$ for any $1\leq i\leq 2d$, it can be seen that $S_A$ and $S_B$ are invertible.}
		Then $\{S_A^{-1/2}(\vec{\alpha}_i^{-1} A_i)S_A^{-1/2}\}$ and $\{S_B^{-1/2}\vec{\alpha}_j^{-1} B_j S_B^{-1/2}\}$
		are measurements as well;
		furthermore, they generate $P$ upon measuring the state $(S_A^{1/2}\otimes S_B^{1/2})\rho(S_A^{1/2}\otimes S_B^{1/2})$.
		{According to our discussions in the previous section, this} implies that the state $(S_A^{1/2}\otimes S_B^{1/2})\rho(S_A^{1/2}\otimes S_B^{1/2})$ is pure,
		hence $\rho$ is also pure.
		By the spectrum of the reduce states of $\rho$, we conclude that $\rho$ is a bipartite pure state with Schmidt coefficients $\sqrt{\lambda_1},\cdots,\sqrt{\lambda_d}$.
	\end{proof}

{\section{The robustness of our protocols}}

In this section, we will prove that the above protocols are robust to quantum noises and experimental imperfections. We assume that $d\geq 2$ throughout this section. {We first introduce the outline of our proof.}

\subsection{Overview of the Proof}
{
Let $(\rho,\{A_i\},\{B_j\})$ be a configuration in a $d\times d$ bipartite system that generates a $2d\times 2d$ correlation $P'$ (given by $P'_{ij}=\mathrm{tr}(\rho(A_i\otimes B_j))$) that is ``close'' to the target correlation $P$ in Eq.\eqref{correlation:maximal}. The proof proceeds as follows:
\begin{itemize}
	\item
		Lemma \ref{lemm:robust-psdrank} shows that the PSD rank of $P'$ remains to be $d$ in certain proximity of $P$.
		This ensures the invertibility of some operators in the following lemmas.
	\item
		The configuration $(\rho,\{A_i\},\{B_j\})$ induces two PSD factorizations:
		\begin{align}
			P'_{ij}=\mathrm{tr}(A_i\cdot\mathrm{tr}_B(\rho(I_A\otimes B_j)))=\mathrm{tr}(\mathrm{tr}_A(\rho(A_i\otimes I_B))\cdot B_j).
		\end{align}
		Since $\{A_i\}$ and $\{B_j\}$ are quantum measurements, both of the above two PSD factorizations have the property that there is one side on which the operators sum to identity. Lemma \ref{lemm:robust-rankone} and Lemma \ref{lemm:robust-rectify} slightly adjust such PSD factorizations on this special side so that the operators in the new PSD factorizations satisfy some additional properties, and furthermore the new correlation is still close to $P$.
	\item
		Lemma \ref{lemm:robust-rectify-config} combines the two paths of PSD factorization adjustment (one for each party) into a single adjustment directly to the original configuration $(\rho,\{A_i\},\{B_j\})$.
		The measurement operators $\tilde{A}_i$ and $\tilde{B}_j$ in the resulting configuration $(\rho,\{\tilde{A}_i\},\{\tilde{B}_j\})$
		are close to the original measurement operators $A_i$ and $B_j$ respectively and have several characterizations
		that enable the state $\rho$ to be determined.
	\item
		Finally, Lemma \ref{lemm:robust-sdp} formulates the robust state certification problem as a semidefinite programming (SDP) problem.
		This is possible since some of the measurement operators have the characterizations given by the previous step.
		The dual problem of this SDP has an explicitly constructed feasible point,
		whose dual objective function value is the desired robustness result.
\end{itemize}
}

{\subsection{Proof for the robustness in certifying maximally entangled states}}

Recall that
\begin{align*}
	P= &
	\frac{1}{d(d+1)^2}\begin{pmatrix}
		d^2 I_d & \vec{e}\vec{e}^T \\
		\vec{e}\vec{e}^T & I_d \\
	\end{pmatrix}.
\end{align*}
{Let $D(P',P)$ be the distance between the correlations $P$ and $P'$, which is defined as $D(P',P)=\max_{i,j}|P'_{ij}-P_{ij}|$. The following conclusion will be useful to us.}

\begin{lemm}
	\label{lemm:robust-psdrank}
	Let $P'$ be a $2d\times 2d$ correlation with $D(P',P)\leq\varepsilon$ and PSD rank at most $d$.
	If $\varepsilon<\frac{1}{d^3(d+1)^2}$, then the PSD rank of $P'$ is $d$.
\end{lemm}
\begin{proof}
	It suffices to examine the {lower right} $d\times d$ block of $P'$, and we suppose this block has PSD rank {$r<d$}.
	Let $\{C_i\}$, $\{D_j\}$ be a PSD factorization of this block, where $\sum_{j=d+1}^{2d}D_j=I_r$. {This is always possible, as otherwise we can adjust the operators with an invertible matrix $H$ like in Lemma \ref{correlation:scsd}.}
	Then for any $i=d+1,\cdots,2d$, it holds that
	{
	\begin{align*}
		\mathrm{tr}(C_i)=\mathrm{tr}(C_i I_r)=\sum_{j=d+1}^{2d}\mathrm{tr}(C_i D_j)=\sum_{j=d+1}^{2d}P'_{ij} \leq \frac{1}{d(d+1)^2}+d \varepsilon,
	\end{align*}
	}
    which gives that $\mathrm{tr}(C_i^2)\leq \left(\frac{1}{d(d+1)^2}+d\varepsilon\right)^2$.
	On the other hand, there exists some $j\in[2d]-[d]$ such that $\mathrm{tr}(D_j)\leq r/d$ since
	$D_j\geq 0$ and $\sum_{j=d+1}^{2d}\mathrm{tr}(D_j)=\mathrm{tr}(I_r)=r$.
	For this same $j$ we have that $\mathrm{tr}(D_j^2)\leq (r/d)^2$.
	By the Cauchy-Schwarz inequality we obtain that
	\begin{align*}
		\left(\frac{1}{d(d+1)^2}-\varepsilon\right)^2\leq (P'_{jj})^{{2}}=\mathrm{tr}(C_j D_j)^2\leq {\mathrm{tr}(C_j^2)\mathrm{tr}(D_j^2)},
	\end{align*}
	which implies
	\begin{align}
		\mathrm{tr}(C_j^2)\geq \left(\frac{d}{r}\right)^2\left(\frac{1}{d(d+1)^2}-\varepsilon\right)^2.
	\end{align}
	If $r\leq d-1$ and $\varepsilon<\frac{1}{d^3(d+1)^2}$, we have that
	\begin{align}
		\left(\frac{d}{r}\right)^2\left(\frac{1}{d(d+1)^2}-\varepsilon\right)^2>\left(\frac{1}{d(d+1)^2}+d\varepsilon\right)^2,
	\end{align}
	which is a contradiction.
	Therefore, the PSD rank of $P'$ is $d$.
\end{proof}

{The proof of Lemma \ref{lemm:robust-psdrank} also implies that when the condition for $\varepsilon$ is satisfied,
both the upper left and the lower right $d\times d$ blocks of $P'$ have PSD rank $d$ (the analysis needs to be repeated with slight modification for the upper left block of $P'$).}


{Below we show that if $\{C_i\}$ and $\{D_j\}$ are a PSD factorization of a correlation $P'$ close to $P$, and $\{C_i\}$ forms a quantum measurement, then we can adjust each $C_i$ slightly to $C'_i$ such that they have rank 1, and the resulting correlation is still close to $P$.}

\begin{lemm}
	\label{lemm:robust-rankone}
	Suppose $P'$ is a $2d\times 2d$ correlation with $D(P',P)\leq\varepsilon$,
	{where $\varepsilon<\frac{1}{8d(d+1)^5}$}.
	Let $\{C_i\}$, $\{D_j\}$ be a PSD factorization of $P'$ such that $\sum_{i=1}^{2d}C_i=I_d$.
	{Then there exists a $2d\times 2d$ correlation $P''$ with $D(P,P'')\leq {18}(d+1)^4\varepsilon$}
	such that $P''$ admits a PSD factorization $\{C'_i\}$, $\{D_j\}$, and
	\begin{itemize}
		\item
			$\forall i\in [2d]$, $\lVert C'_i-C_i\rVert_{\text{op}}\leq 16(d+1)^6\varepsilon$, {where $\lVert \cdot\rVert_{\text{op}}$ denotes the operator norm};
		\item
			All the operators $C'_i$ are rank $1$;
		\item
			$\sum_{i=1}^{2d}C'_i=I_d$.
	\end{itemize}
\end{lemm}
\begin{proof}
    {We will adjust all the $2d$ operators $C_i$ by two steps, and the first step adjusts \{$C_{d+1},\cdots C_{2d}$\}.} For the given PSD factorization $\{C_i\}$ and $\{D_j\}$ of $P'$,
	we have
	\begin{align}
		\mathrm{rank}\left(\sum_{i=1}^{d}C_{d+i}\right)=\mathrm{rank}\left(\sum_{j=1}^{d}{D_{d+j}}\right)=d
	\end{align}
	by Lemma \ref{lemm:robust-psdrank}.
	Let
	\begin{align*}
		\tilde{C}_i= & \left(\sum_{i=1}^{d}C_{d+i}\right)^{-1/2}C_i\left(\sum_{i=1}^{d}C_{d+i}\right)^{-1/2}, \\
		\tilde{D}_j= & \left(\sum_{i=1}^{d}C_{d+i}\right)^{1/2}D_j\left(\sum_{i=1}^{d}C_{d+i}\right)^{1/2}.
	\end{align*}
	Then $\{\tilde{C}_i\}$ and $\{\tilde{D}_j\}$ are equivalent to $\{C_i\}$ and $\{D_j\}$, {and it holds that $\sum_{i=1}^{d} \tilde{C}_{d+i}=I_d$}.
	Consider $\tilde{D}_{d+j}$ where $j\in[d]$.
	Firstly we have $\mathrm{tr}(\tilde{D}_{d+j})=\sum_{{i}=1}^d\mathrm{tr}(\tilde{C}_{d+i}\tilde{D}_{d+j})\leq \frac{1}{d(d+1)^2}+d\varepsilon$.
	Thus
	\begin{align*}
		\left(\frac{1}{d(d+1)^2}-\varepsilon\right)^2 \leq & \mathrm{tr}(\tilde{C}_{d+j} \tilde{D}_{d+j})^2 \\
		\leq & \mathrm{tr}(\tilde{C}_{d+j}^2)\mathrm{tr}(\tilde{D}_{d+j}^2) \\
		\leq & \mathrm{tr}(\tilde{C}_{d+j}^2)\mathrm{tr}(\tilde{D}_{d+j})^2 \\
		\leq & \left(\frac{1}{d(d+1)^2}+d\varepsilon\right)^2\mathrm{tr}(\tilde{C}_{d+j}^2),
	\end{align*}
	giving
	\begin{align}
		\mathrm{tr}(\tilde{C}_{d+j}^2)\geq \left(\frac{\frac{1}{d(d+1)^2}-\varepsilon}{\frac{1}{d(d+1)^2}+d\varepsilon}\right)^2\geq 1-2d(d+1)^3\varepsilon.
	\end{align}
	We then have
	\begin{align}
		\sum_{i=1}^{d}(\mathrm{tr}(\tilde{C}_{d+i})-\mathrm{tr}(\tilde{C}_{d+i}^2))=d-\sum_{i=1}^{d}\mathrm{tr}(\tilde{C}_{d+i}^2)\leq 2d^2(d+1)^3\varepsilon,
	\end{align}
    {where we have utilized the relation $\sum_{i=1}^{d} \tilde{C}_{d+i}=I_d$.}
	Since $\sum_{i=1}^{d} \tilde{C}_{d+i}=I_d$, we also have $\mathrm{tr}(\tilde{C}_{d+i})\geq\mathrm{tr}(\tilde{C}_{d+i}^2)$,
	thus for $\forall i\in[d]$, it holds that
	\begin{align*}
		\mathrm{tr}(\tilde{C}_{d+i})-\mathrm{tr}(\tilde{C}_{d+i}^2)\leq 2d^2(d+1)^3\varepsilon.
	\end{align*}
	For any eigenvalue $\mu$ of $\tilde{C}_{d+i}$, we have $\mu\in[0,1]$,
	and
	\begin{align}
		\frac{1}{2}\min\{\mu,1-\mu\}\leq \mu(1-\mu)\leq \mathrm{tr}(\tilde{C}_{d+i})-\mathrm{tr}(\tilde{C}_{d+i}^2)\leq 2d^2(d+1)^3\varepsilon,
	\end{align}
	hence either $\mu\in[0,4d^2(d+1)^3\varepsilon]$ or $\mu\in[1-4d^2(d+1)^3\varepsilon,1]$. {This fact can be seen as below. Suppose} none of $\mu$ belongs to the latter case, then
	\begin{align}
		\mathrm{tr}(\tilde{C}_{d+i}^2)\leq \mathrm{tr}(\tilde{C}_{d+i})\leq 4d^3(d+1)^3\varepsilon<1-2d(d+1)^3\varepsilon,
	\end{align}
	which is a contradiction.
	On the other hand, if $\tilde{C}_{d+i}$ has two eigenvalues that belong to the latter case, then $\mathrm{tr}(\tilde{C}_{d+i})>1.5$,
	hence
	\begin{align*}
		d=\mathrm{tr}(\tilde{C}_{d+i})+\sum_{i'\neq i}\mathrm{tr}(\tilde{C}_{d+i'})>1.5+\sum_{i'\neq i}\mathrm{tr}(\tilde{C}_{d+i'}^2)\geq d+0.5-2(d-1)d(d+1)^3\varepsilon,
	\end{align*}
	which is again a contradiction.
	Therefore, each $\tilde{C}_{d+i}$ has exactly one eigenvalue within $[1-4d^2(d+1)^3\varepsilon,1]$ (i.e., close to $1$) and the rest in $[0,4d^2(d+1)^3\varepsilon]$ (i.e., close to $0$).
	Now suppose $v_i$ is the unit vector corresponding to the eigenvalue of $\tilde{C}_{d+i}$ that is close to $1$.
	We have thus established that $\lVert \tilde{C}_{d+i}-v_iv_i^{\dagger}\rVert_{\text{op}}\leq 4d^2(d+1)^3\varepsilon$ for all $i\in[d]$.
	
	Consider $X=\sum_{i=1}^{d}v_i v_i^{\dagger}$;
	since it holds that
	\begin{align}
		\left\lVert X-I\right\rVert_{\text{op}}\leq \sum_{i=1}^{d}\left\lVert v_iv_i^{\dagger} - \tilde{C}_{d+i}\right\rVert_{\text{op}}\leq 4d^3(d+1)^3\varepsilon,
	\end{align}
	the operator $X$ is invertible.
	Therefore we may define $\hat{C}_{d+i}=X^{-1/2}v_i v_i^{\dagger} X^{-1/2}$
	so that $\sum_{i=1}^{d}\hat{C}_{d+i}=I_d$.
	{It can be verified} that
	\begin{align}
		(1-4d^3(d+1)^3\varepsilon)I\leq X^{-1/2}\leq (1+4d^3(d+1)^3\varepsilon)I.
	\end{align}
	This implies that
	\begin{align*}
		\left\lVert\tilde{C}_{d+i}-\hat{C}_{d+i}\right\rVert_{\text{op}}
		\leq & \left\lVert \tilde{C}_{d+i}-v_iv_i^{\dagger}\right\rVert_{\text{op}}+\left\lVert v_iv_i^{\dagger}-X^{-1/2}v_iv_i^{\dagger} X^{-1/2}\right\rVert_{\text{op}} \\
		= & \left\lVert \tilde{C}_{d+i}-v_iv_i^{\dagger}\right\rVert_{\text{op}}+\left\lVert (I-X^{-1/2}) v_iv_i^{\dagger}X^{-1/2}-v_iv_i^{\dagger}(X^{-1/2}-I)\right\rVert_{\text{op}} \\
		\leq & \left\lVert \tilde{C}_{d+i}-v_iv_i^{\dagger}\right\rVert_{\text{op}}+\left\lVert I-X^{-1/2}\right\rVert_{\text{op}}\left(1+\left\lVert X^{-1/2}\right\rVert_{\text{op}}\right) \\
		\leq & 4d^2(d+1)^3\varepsilon+12d^3(d+1)^3\varepsilon \\
		\leq & 16d^3(d+1)^3\varepsilon,
	\end{align*}
	which upon reversing the equivalence transformation gives
	\begin{align*}
		& \left\lVert C_{d+i}-\left(\sum_{i=1}^{d}C_{d+i}\right)^{1/2}\hat{C}_{d+i}\left(\sum_{i=1}^{d}C_{d+i}\right)^{1/2}\right\rVert_{\text{op}} \\
		= & \left\lVert\left(\sum_{i=1}^{d}C_{d+i}\right)^{1/2}\left(\tilde{C}_{d+i}-\hat{C}_{d+i}\right)\left(\sum_{i=1}^{d}C_{d+i}\right)^{1/2}\right\rVert_{\text{op}} \\
		\leq & \left\lVert\sum_{i=1}^{d}C_{d+i}\right\rVert_{\text{op}}\left\lVert\tilde{C}_{d+i}-\hat{C}_{d+i}\right\rVert_{\text{op}} \\
		\leq & 16d^3(d+1)^3\varepsilon,
	\end{align*}
    {where we have utilized the fact that $\sum_{i=1}^{2d}C_i=I_d$, and thus $\sum_{i=1}^{d}C_{d+i}\leq I_d$}.
	Let $P''$ be a $2d\times 2d$ correlation given by
	\begin{align}
		P''_{ij}=\begin{cases}
			\mathrm{tr}(\tilde{C}_i \tilde{D}_j), & i\in[d] \\
			\mathrm{tr}(\hat{C}_i \tilde{D}_j), & {i\in[2d]-[d]}
		\end{cases}.
	\end{align}
	For $j\in[2d]$, {since $\sum_{i=1}^{d}\hat{C}_{d+i}=I_d$}, we have
	\begin{align}
		\mathrm{tr}(\tilde{D}_j)=\sum_{i=1}^{d}\mathrm{tr}(\tilde{C}_{d+i}\tilde{D}_j)\leq \frac{1}{(d+1)^2}+d\varepsilon\leq {\frac{17}{16(d+1)^{2}}},
	\end{align}
	hence by H\"{o}lder's inequality we have
	\begin{align}
		D(P'',P')= & \max_{i\in[d],j\in[2d]}\left|\mathrm{tr}((\tilde{C}_{d+i}-\hat{C}_{d+i})\tilde{D}_j)\right| \\
		\leq & 16d^3(d+1)^3\varepsilon\left(\max_{j\in[2d]}\mathrm{tr}(\tilde{D}_j)\right) \\
		\leq & {17d^3(d+1)}\varepsilon,
	\end{align}
	which also implies that {$D(P'',P)\leq{17(d+1)^4}\varepsilon$}.
	We now apply the inverse of the equivalence transformation on $\hat{C}_{d+i}$ to obtain $\{C'_{d+i}\}$, i.e.,
    \begin{align*}
    {C'_{d+i}=\left(\sum_{i=1}^{d}C_{d+i}\right)^{1/2}\hat{C}_{d+i}\left(\sum_{i=1}^{d}C_{d+i}\right)^{1/2}.}
    \end{align*}
	Then it could be verified that $C'_{d+i}$ is rank 1 for any $i\in[d]$ and $\sum_{i\in[d]}C'_{d+i}=\sum_{i\in[d]}C_{d+i}$.

	{In the second step of adjusting $C_i$, we repeat the whole process for \{$C_1,\cdots,C_d$\} and yield} a final correlation $P'''$ such that {$D(P''',P)\leq$ }{$18(d+1)^4\varepsilon$} and $P'''$ admits the desired PSD factorization.
\end{proof}

{The next lemma slightly adjusts the above $C'_i$ further without changing the generated correlation much such that more desirable mathematical structures show up, which will be useful later}.

\begin{lemm}
	\label{lemm:robust-rectify}
	Suppose $P'$ is a correlation with $D(P',P)\leq\varepsilon$
	where $\varepsilon<{(18)^{-8}(d(d+1)^2)^{-4}}$.
	Furthermore assuming that $P'$ admits a rank $d$ PSD factorization $\{C_i\}$, $\{D_j\}$
	where all operators $\{C_i\}$ are of rank 1 and $\sum_{i=1}^{2d}C_i=I_d$.
	Then there exists a correlation $P''$ with $D(P,P'')\leq{64d^{-1/2}\varepsilon^{1/8}}$
	that admits a PSD factorization $\{C'_i\}$ and $\{D_j\}$ such that:
	\begin{enumerate}
		\item
			$\forall i\in[2d], \lVert C_i-C'_i\rVert_{\text{op}}\leq {64d^{-1/2}(d+1)\varepsilon^{1/8}}$;
		\item
			$\forall i\in[2d], \mathrm{rank}(C'_i)=1$;
		\item
            {$\sum_{i=1}^dC'_{i}=\frac{d}{d+1}I_d$;}
		\item
		{$\sum_{i=1}^dC'_{d+i}=\frac{1}{d+1}I_d$;}
		\item
			{$\forall i\in[d], \mathrm{tr}(C'_i C'_{d+1})=\frac{1}{(d+1)^2}$.}
	\end{enumerate}
\end{lemm}
\begin{proof}
	{Firstly, by Lemma \ref{lemm:robust-psdrank} we know that the PSD rank of $P'$ is $d$,
	then it can be verified that the vectors corresponding to the operators $C_i$ {for all} $i\in[d]$ are linearly independent; a similar conclusion holds for all $C_{d+i}$ where $i\in[d]$.}
	Therefore, the operator $H$ given by $H=\left(\sum_{i=1}^{d}C_i\right)^{-1/2}$ is positive and invertible.
	Let $\tilde{C}_i=HC_i H^{\dagger}$ and $\tilde{D}_j=(H^{\dagger})^{-1}D_j H^{-1}$;
	since the operators $\tilde{C}_i$ form a rank 1 resolution of identity, {i.e., $\sum_{i=1}^{d}\tilde{C}_i=I_d$, then for any $i\in[d]$} there are orthonormal vectors $e_i$
	such that $\tilde{C}_i=e_i e_i^{\dagger}$.
	Thus, for any $i\in[d]$ we have
	\begin{align*}
		\left\lVert\tilde{D}_i-\frac{d}{(d+1)^2}e_i e_i^{\dagger}\right\rVert_{\text{op}}^2
		\leq & \left\lVert\tilde{D}_i-\frac{d}{(d+1)^2}e_i e_i^{\dagger}\right\rVert_{\text{F}}^2 \\
		= & {\mathrm{tr}(\tilde{D}_i^2)-\frac{2d}{(d+1)^2}\mathrm{tr}(\tilde{D}_i e_ie_i^{\dagger})+\frac{d^2}{(d+1)^4}} \\
		\leq & {\mathrm{tr}(\tilde{D}_i)^2-\frac{2d}{(d+1)^2}\mathrm{tr}(\tilde{D}_i e_ie_i^{\dagger})+\frac{d^2}{(d+1)^4}} \\
		\leq & {\left(\frac{d}{(d+1)^2}+d\varepsilon\right)^2-\frac{2d}{(d+1)^2}\left(\frac{d}{(d+1)^2}-\varepsilon\right)+\frac{d^2}{(d+1)^4}} \\
		= & {\frac{2d}{d+1}\varepsilon+d^2\varepsilon^2} \\
		\leq & 3\varepsilon.
	\end{align*}
	{where $\lVert \cdot \rVert_{\text{F}}$ denotes the Frobenius norm, and we have utilized the fact that $\mathrm{tr}(\tilde{D}_i)=\sum_{j=1}^d\mathrm{tr}(\tilde{C}_j\tilde{D}_i)$.}
	Thus
	\begin{align*}
		\left\lVert\sum_{j=1}^{d}\tilde{D}_j-\frac{d}{(d+1)^2}I_d\right\rVert_{\text{op}}\leq 2d\sqrt{\varepsilon}.
	\end{align*}
	For any $i\in[d]$, we have
	\begin{align*}
		\left|\mathrm{tr}(\tilde{C}_{d+i})-\frac{1}{d}\right|=
		& \left|\frac{(d+1)^2}{d}\mathrm{tr}\left(\tilde{C}_{d+i}\cdot \frac{d}{(d+1)^2}I_d\right)-\frac{1}{d}\right| \\
		= & \left|\frac{(d+1)^2}{d}\mathrm{tr}\left(\tilde{C}_{d+i}\left({\frac{d}{(d+1)^2}I_d}-\sum_{j=1}^{d}\tilde{D}_j+\sum_{j=1}^{d}\tilde{D}_j\right)\right)-\frac{1}{d}\right| \\
		= & \left|\frac{(d+1)^2}{d}\sum_{j=1}^{d}\left({\mathrm{tr}(\tilde{C}_{d+i}\tilde{D}_j)-\frac{1}{d(d+1)^2}}\right)+\frac{(d+1)^2}{d}\mathrm{tr}\left(\tilde{C}_{d+i}\left(\frac{d}{(d+1)^2}I_d-\sum_{j=1}^{d}\tilde{D}_j\right)\right)\right| \\
		\leq & \sum_{j=1}^{d}\frac{(d+1)^2}{d}\left|P'_{d+i,j}-\frac{1}{d(d+1)^2}\right|+\frac{(d+1)^2}{d}\mathrm{tr}(\tilde{C}_{d+i})\left\lVert \frac{d}{(d+1)^2}I_d-\sum_{j=1}^{d}\tilde{D}_j\right\rVert_{\text{op}} \\
		\leq & (d+1)^2\varepsilon+{2(d+1)^2\sqrt{\varepsilon}\mathrm{tr}(\tilde{C}_{d+i})}.
	\end{align*}
	Since $\mathrm{tr}(\tilde{C}_{d+i})\leq \left|\mathrm{tr}(\tilde{C}_{d+i})-\frac{1}{d}\right|+\frac{1}{d}$, we have
	\begin{align*}
		\left|\mathrm{tr}(\tilde{C}_{d+i})-\frac{1}{d}\right|\leq 2(d+1)^2\sqrt{\varepsilon}\left|\mathrm{tr}(\tilde{C}_{d+i})-\frac{1}{d}\right|+(d+1)^2\varepsilon+\frac{2(d+1)^2}{d}\varepsilon^{1/2}.
	\end{align*}
	Thus
	\begin{align}\label{eq:cd+i}
		\left|\mathrm{tr}(\tilde{C}_{d+i})-\frac{1}{d}\right|\leq {4(d+1)\sqrt{\varepsilon}}.
	\end{align}
	For $j\in[d]$, we have
	\begin{align}\label{eq:dd+j}
		\left|\mathrm{tr}(\tilde{D}_{{d+j}})-\frac{1}{(d+1)^2}\right|\leq\sum_{i=1}^{d}\left|\mathrm{tr}(\tilde{D}_{{d+j}}\tilde{C}_i)-\frac{1}{d(d+1)^2}\right|\leq d\varepsilon,
	\end{align}
    {where we have utilized the relation $\sum_{i=1}^{d}\tilde{C}_i=I_d$.}

	Below we prove that $\left\lVert\sum_{i=1}^{d}\tilde{C}_{d+i}-\frac{1}{d}I_d\right\rVert_{\text{op}}\leq$ {$5d^{1/2}(d+1)^{1/2}\varepsilon^{1/4}$}. For this, note that the supports of $\tilde{C}_{d+i}$ and $\tilde{C}_{d+i'}$ span a subspace $V_{ii'}$ of dimension {at most} $2$, whose orthogonal {projection} will be denoted as {$\Pi_{ii'}:\mathbb{C}^{d}\to V_{ii'}$}.
	Let $\hat{D}_{d+i}=\Pi_{ii'}\tilde{D}_{d+i}\Pi_{ii'}^{\dagger}$ and $\hat{D}_{d+i'}=\Pi_{ii'}\tilde{D}_{d+i'}\Pi_{ii'}^{\dagger}$.
	Then in $V_{ii'}$ we have {(the operators in $V_{ii'}$ are understood as $2\times 2$ matrices)}
	\begin{align*}
	& \left\lVert(\tilde{C}_{d+i}+\tilde{C}_{d+i'})^{1/2}\hat{D}_{d+i}(\tilde{C}_{d+i}+\tilde{C}_{d+i'})^{1/2}-\frac{1}{d(d+1)^2}(\tilde{C}_{d+i}+\tilde{C}_{d+i'})^{-1/2}\tilde{C}_{d+i}(\tilde{C}_{d+i}+\tilde{C}_{d+i'})^{-1/2}\right\rVert_{\text{F}}^2 \\
	{=} & \mathrm{tr}\left(\left((\tilde{C}_{d+i}+\tilde{C}_{d+i'})\hat{D}_{d+i}\right)^2\right)-\frac{2}{d(d+1)^2}\mathrm{tr}(\tilde{C}_{d+i}\hat{D}_{d+i})+\frac{1}{d^2(d+1)^4} \\
	\leq & \left(\frac{1}{d(d+1)^2}+2\varepsilon\right)^2-\frac{2}{d(d+1)^2}\left(\frac{1}{d(d+1)^2}-\varepsilon\right)+\frac{1}{d^2(d+1)^4} \\
	= & \frac{6}{d(d+1)^2}\varepsilon+4\varepsilon^2,
\end{align*}
where {the first equation is the fact that $\lVert A\rVert_{\text{F}}^2=\mathrm{tr}(AA^\dagger)$ for any matrix $A$, and} the first inequality holds due to the fact that $(\tilde{C}_{d+i}+\tilde{C}_{d+i'})^{-1/2}\tilde{C}_{d+i}(\tilde{C}_{d+i}+\tilde{C}_{d+i'})^{-1/2}$ and $(\tilde{C}_{d+i}+\tilde{C}_{d+i'})^{-1/2}\tilde{C}_{d+i'}(\tilde{C}_{d+i}+\tilde{C}_{d+i'})^{-1/2}$ form a rank 1 orthogonal resolution of {$I_{V_{ii'}}$}, {as they are rank-1 and sum to the identity}. When $\varepsilon$ is sufficiently small, we have that
\begin{align*}
	& \left\lVert(\tilde{C}_{d+i}+\tilde{C}_{d+i'})^{1/2}\hat{D}_{d+i}(\tilde{C}_{d+i}+\tilde{C}_{d+i'})^{1/2}-\frac{1}{d(d+1)^2}(\tilde{C}_{d+i}+\tilde{C}_{d+i'})^{-1/2}\tilde{C}_{d+i}(\tilde{C}_{d+i}+\tilde{C}_{d+i'})^{-1/2}\right\rVert_{\text{op}} \\
	\leq & 3d^{-1/2}(d+1)^{-1}\varepsilon^{1/2}.
\end{align*}
Summing a similar inequality for $\hat{D}_{d+i'}$ gives
\begin{align*}
	\left\lVert(\tilde{C}_{d+i}+\tilde{C}_{d+i'})^{1/2}(\hat{D}_{d+i}+\hat{D}_{d+i'})(\tilde{C}_{d+i}+\tilde{C}_{d+i'})^{1/2}-\frac{1}{d(d+1)^2}{I_{V_{ii'}}}\right\rVert_{\text{op}}\leq 6d^{-1/2}(d+1)^{-1}\varepsilon^{1/2}.
\end{align*}
Therefore
\begin{align}\label{eq:optotrace}
    \left\lVert(\hat{D}_{d+i}+\hat{D}_{d+i'})-\frac{1}{d(d+1)^2}(\tilde{C}_{d+i}+\tilde{C}_{d+i'})^{-1}\right\rVert_{\text{op}}\leq 6d^{-1/2}(d+1)^{-1}\varepsilon^{1/2}\lVert(\tilde{C}_{d+i}+\tilde{C}_{d+i'})^{-1}\rVert_{\text{op}}.
\end{align}
Since $\tilde{C}_{d+i}$ and $\tilde{C}_{d+i'}$ are both rank-1, it can be verified that in the subspace $V_{ii'}$ we have
\begin{align*}
	(\tilde{C}_{d+i}+\tilde{C}_{d+i'})^{-1}=\frac{1}{\mathrm{tr}(\tilde{C}_{d+i})\mathrm{tr}(\tilde{C}_{d+i'})-\mathrm{tr}(\tilde{C}_{d+i}\tilde{C}_{d+i'})}\left((\mathrm{tr}(\tilde{C}_{d+i})+\mathrm{tr}(\tilde{C}_{d+i'}))I_{V_{ii'}}-(\tilde{C}_{d+i}+\tilde{C}_{d+i'})\right);
\end{align*}
in particular,
\begin{align*}
	\mathrm{tr}\left((\tilde{C}_{d+i}+\tilde{C}_{d+i'})^{-1}\right)=\frac{\mathrm{tr}(\tilde{C}_{d+i})+\mathrm{tr}(\tilde{C}_{d+i'})}{\mathrm{tr}(\tilde{C}_{d+i})\mathrm{tr}(\tilde{C}_{d+i'})-\mathrm{tr}(\tilde{C}_{d+i} \tilde{C}_{d+i'})}.
\end{align*}
We use this to estimate $\mathrm{tr}(\hat{D}_{d+i}+\hat{D}_{d+i'})$:
\begin{align*}
	& 2\left(\frac{1}{(d+1)^2}+d\varepsilon\right) \\
	\geq & \mathrm{tr}(\hat{D}_{d+i}+\hat{D}_{d+i'}) \\
	= & \frac{1}{d(d+1)^2}\mathrm{tr}\left((\tilde{C}_{d+i}+\tilde{C}_{d+i'})^{-1}\right)+\mathrm{tr}\left((\hat{D}_{d+i}+\hat{D}_{d+i'})-\frac{1}{d(d+1)^2}(\tilde{C}_{d+i}+\tilde{C}_{d+i'})^{-1}\right) \\
	\geq & \left(\frac{1}{d(d+1)^2}-12d^{-1/2}(d+1)^{-1}\varepsilon^{1/2}\right)\mathrm{tr}\left((\tilde{C}_{d+i}+\tilde{C}_{d+i'})^{-1}\right) \\
	\geq & \left(\frac{1}{d(d+1)^2}-12d^{-1/2}(d+1)^{-1}\varepsilon^{1/2}\right)\frac{2\left(\frac{1}{d}-{4(d+1)}\varepsilon^{1/2}\right)}{\mathrm{tr}(\tilde{C}_{d+i})\mathrm{tr}(\tilde{C}_{d+i'})-\mathrm{tr}(\tilde{C}_{d+i}\tilde{C}_{d+i'})},
\end{align*}
{where the first inequality comes from Eq.\eqref{eq:dd+j}, the second inequality uses the relation in Eq.\eqref{eq:optotrace} and the fact that $\lVert A\rVert_{\text{op}}\leq\mathrm{tr}(A)$ for any positive matrix $A$, and the third inequality comes from Eq.\eqref{eq:cd+i}.}
When $\varepsilon$ is sufficiently small, this gives
\begin{align*}
	\mathrm{tr}(\tilde{C}_{d+i}\tilde{C}_{d+i'})\leq & \left(\frac{1}{d}+{4(d+1)}\varepsilon^{1/2}\right)^2-\left(\frac{\frac{1}{d}-{4(d+1)}\varepsilon^{1/2}}{\frac{1}{(d+1)^2}+d\varepsilon}\right)\left(\frac{1}{d(d+1)^2}-12d^{-1/2}(d+1)^{-1}\varepsilon^{1/2}\right) \\
	\leq & {23\frac{d+1}{d}}\varepsilon^{1/2},
\end{align*}
{where the first inequality has utilized Eq.\eqref{eq:cd+i}}.
Thus
\begin{align*}
	& \left\lVert\sum_{i=1}^{d}\tilde{C}_{d+i}-\frac{1}{d}I_d\right\rVert_{\text{F}}^2 \\
	= & \sum_{i=1}^{d}\mathrm{tr}(\tilde{C}_{d+i}^2)+\sum_{i\neq i'}\mathrm{tr}(\tilde{C}_{d+i}\tilde{C}_{d+i'})-\frac{2}{d}\sum_{i=1}^{d}\mathrm{tr}(\tilde{C}_{d+i})+\frac{1}{d} \\
	\leq & d\left(\frac{1}{d}+{4(d+1)}\varepsilon^{1/2}\right)^{{2}}+d(d-1)\left({23\frac{d+1}{d}}\varepsilon^{1/2}\right)-2\left(\frac{1}{d}-{4(d+1)}\varepsilon^{1/2}\right)+\frac{1}{d} \\
	\leq & {24d(d+1)}\varepsilon^{1/2}.
\end{align*}
which implies that
\begin{align}
	\left\lVert\sum_{i=1}^{d}\tilde{C}_{d+i}-\frac{1}{d}I_d\right\rVert_{\text{op}}\leq {5d^{1/2}(d+1)^{1/2}}\varepsilon^{1/4}.
\end{align}

	Next, {since $\sum_{i=1}^{2d}C_i=I_d$, we have} that
	\begin{align}
		\left\lVert H^2-\frac{d+1}{d}I_d\right\rVert_{\text{op}}= & \left\lVert \sum_{i=1}^{2d}HC_i H^{\dagger}-\frac{d+1}{d}I_d\right\rVert_{\text{op}}  \\
		= & \left\lVert \sum_{i=1}^{2d}\tilde{C}_i-\frac{d+1}{d}I_d\right\rVert_{\text{op}} \\
		= & \left\lVert \sum_{i=d+1}^{2d}\tilde{C}_i-\frac{1}{d}I_d\right\rVert_{\text{op}} \\
		\leq & {5d^{1/2}(d+1)^{1/2}}\varepsilon^{1/4},
	\end{align}
	which, under the assumption that $\varepsilon<{(5)^{-4}(d+1)^{-8}}$, also implies that
	\begin{align}
		\left\lVert \sum_{i=1}^{d}C_i-\frac{d}{d+1}I_d\right\rVert_{\text{op}}= & \left\lVert H^{-2}-\frac{d}{d+1}I_d\right\rVert_{\text{op}} \\
		\leq & {5d^{1/2}(d+1)^{1/2}}\varepsilon^{1/4}, \\
		\left\lVert\sqrt{\frac{d}{d+1}}H-I_d\right\rVert_{\text{op}}\leq & {5d}\varepsilon^{1/4}, \\
		\left\lVert\sqrt{\frac{1}{d+1}}(I-H^{-2})^{-1/2}-I_d\right\rVert_{\text{op}}\leq & {5d^{1/2}(d+1)^{3/2}}\varepsilon^{1/4},
	\end{align}
    {where we have utilized the definition of $H$, i.e., $H=\left(\sum_{i=1}^{d}C_i\right)^{-1/2}$.}
	Consider a $2d\times 2d$ correlation $P^{(1)}$ given by the operators $\hat{C}_i=\frac{d}{d+1}H C_i H$ for $i\in[d]$,
	$\hat{C}_i=\frac{1}{d+1}(I-H^{-2})^{-1/2}C_i(I-H^{-2})^{-1/2}$ for $i\in[2d]-[d]$, and $D_j$ for $j\in[2d]$. {Then it holds that $\sum_{i=1}^d\hat{C}_i=\frac{d}{d+1}I_d$ and $\sum_{i=1}^d\hat{C}_{i+d}=\frac{1}{d+1}I_d$.}
	We first bound the trace of $D_j$ for $j\in[d]$ as
	\begin{align}
		\mathrm{tr}(D_j)=\sum_{i=1}^{2d}\mathrm{tr}(C_i D_j)\leq \frac{1}{d+1}+{2}d\varepsilon,
	\end{align}
	and for $j\in[2d]-[d]$ as
	\begin{align}
		\mathrm{tr}(D_j)\leq d^{-1}(d+1)^{-1}+{2}d\varepsilon.
	\end{align}
	Then for $i\in[d]$ and $j\in [2d]$, since $\lVert C_i\rVert_{\text{op}}\leq {1}$ {(recall that $\sum_{i=1}^{2d}C_i=I_d$)},
	we have
	\begin{align*}
		\left|P^{(1)}_{ij}-P'_{ij}\right|= & \left|\mathrm{tr}\left(\left(\frac{d}{d+1}HC_i H-C_i\right)D_j\right)\right| \\
		= & \left|\mathrm{tr}\left(\left(\left(\sqrt{\frac{d}{d+1}}H\right)C_i\left(\sqrt{\frac{d}{d+1}}H-I_d\right)+\left(\sqrt{\frac{d}{d+1}}H-I_d\right)C_i\right)D_j\right)\right| \\
		\leq & \left\lVert\sqrt{\frac{d}{d+1}}H-I_d\right\rVert_{\text{op}}\lVert C_i\rVert_{\text{op}}\left(1+\left\lVert\sqrt{\frac{d}{d+1}}H\right\rVert_{\text{op}}\right)\mathrm{tr}(D_j) \\
		\leq & {18\varepsilon^{1/4}},
	\end{align*}
	while similarly for $i\in[2d]-[d]$ and $j\in [2d]$, it holds that
	\begin{align*}
		\left|P^{(1)}_{ij}-P'_{ij}\right|= &
		\left|\mathrm{tr}\left(\left(\left(\frac{1}{d{+}1}(I_d-H^{-2})^{-1/2}C_i(I_d-H^{-2})^{-1/2}\right)-C_i\right)D_j\right)\right| \\
		\leq & \left\lVert\sqrt{\frac{1}{d+1}}(I_d-H^{-2})^{-1/2}-I_d\right\rVert_{\text{op}}\lVert C_i\rVert_{\text{op}}\left(1+\left\lVert\sqrt{\frac{1}{d+1}}(I_d-H^{-2})^{-1/2}\right\rVert_{\text{op}}\right)\mathrm{tr}(D_j) \\
		\leq & {36\varepsilon^{1/4}},
	\end{align*}
	where we have used that {$\sum_{j=d+1}^{2d}C_i=I_d-H^{-2}$, and thus}
	\begin{align*}
		\lVert C_i\rVert_{\text{op}}\leq \lVert I_d-H^{-2}\rVert_{\text{op}}\leq \frac{1}{d+1}+16d\varepsilon^{1/4}\leq \frac{2}{d+1}
	\end{align*}
	for any $i\in[2d]-[d]$.
	Therefore
	\begin{align}
		D(P^{(1)},P')\leq {36\varepsilon^{1/4}}.
	\end{align}
	{Based on the definitions of $P^{(1)}$ and $P'$, this also shows that for any $i\in[2d]$}
	\begin{align}
		\left\lVert \hat{C}_i-C_i\right\rVert_{\text{op}}\leq {30(d+1)\varepsilon^{1/4}}.
	\end{align}

	To complete the proof, we need to regulate $\mathrm{tr}(C_i C_{d+1})$ so that
	$\mathrm{tr}(C_i C_{d+1})=\frac{1}{(d+1)^2}$ for any $i\in[d]$.
	First for any $i\in[d]$, we write
	\begin{align*}
		\hat{C}_i= & \frac{d}{d+1}u_i u_i^{\dagger}, \\
		\hat{C}_{d+i}= & \frac{1}{d+1}w_i w_i^{\dagger},
	\end{align*}
	where $u_i, w_i$ are unit vectors in $\mathbb{C}^d$.
	This is possible since $\hat{C}_1,\cdots,\hat{C}_d$ form a rank $1$ orthogonal resolution of $\frac{d}{d+1}I_d$
	and $\hat{C}_{d+1},\cdots, \hat{C}_{2d}$ form a rank $1$ orthogonal resolution of $\frac{1}{d+1}I_d$.
	For $i\in[d]$, we also have
	\begin{align*}
		& \left\lVert D_i-\frac{1}{d+1}u_i u_i^{\dagger}\right\rVert_{\text{F}}^2 \\
		= & \lVert D_i\rVert_{\text{F}}^2-\frac{2}{d+1}\mathrm{tr}((u_i u_i^{\dagger})D_i)+\frac{1}{(d+1)^2} \\
		= & \lVert D_i\rVert_{\text{F}}^2-\frac{2}{d}\mathrm{tr}({\hat{C}_i} D_i)+\frac{1}{(d+1)^2} \\
		\leq & \left({\frac{1}{d+1}+2d\varepsilon}\right)^2-\frac{2}{d}\left(\frac{d}{(d+1)^2}-\varepsilon-{18\varepsilon^{1/4}}\right)+\frac{1}{(d+1)^2} \\
		\leq & {42d^{-1}\varepsilon^{1/4}}.
	\end{align*}
	which implies that
	\begin{align}
		\left\lVert D_i-\frac{1}{d+1}u_i u_i^{\dagger}\right\rVert_{\text{op}}\leq {7d^{-1/2}\varepsilon^{1/8}}.
	\end{align}
	Hence
	\begin{align*}
		& \left|\frac{d}{(d+1)^2}|u_i^{\dagger} w_1|^2-\frac{1}{(d+1)^2}\right| \\
		= & \left|\mathrm{tr}(\hat{C}_{d+1}\hat{C}_i)-\frac{1}{(d+1)^2}\right| \\
		= & \left|\mathrm{tr}\left(\hat{C}_{d+1}\left(\left(\frac{d}{d+1}u_i u_i^{\dagger}-d D_i\right)+d D_i\right)\right)-\frac{1}{(d+1)^2}\right| \\
		\leq & \left|d\mathrm{tr}(\hat{C}_{d+1} D_i)-\frac{1}{(d+1)^2}\right|+\left\lVert\frac{d}{d+1}u_i u_i^{\dagger}-d D_i\right\rVert_{\text{op}}\mathrm{tr}(\hat{C}_{d+1}) \\
		\leq & {d\left|P^{(1)}_{d+1,j}-\frac{1}{d(d+1)^2}\right|+{\frac{7d^{1/2}}{d+1}\varepsilon^{1/8}}} \\
		\leq & {d\varepsilon+36d\varepsilon^{1/4}+7d^{-1/2}\varepsilon^{1/8}} \\
		\leq & {10d^{-1/2}\varepsilon^{1/8}}.
	\end{align*}
	Define $\delta_i=u_i^{\dagger} w_1$; then $w_1=\sum_{i=1}^{d}\delta_i u_i$.
	The previous result then reads
	\begin{align}
		\left||\delta_i|^2-\frac{1}{d}\right|\leq {10d^{-3/2}(d+1)^2\varepsilon^{1/8}};
	\end{align}
	in particular, we have $\delta_i\neq 0$ for any $\in[d]$.
	Consider the unit vector ${\tilde{w}}$ given by
	\begin{align}
		{\tilde{w}}=\sum_{i=1}^{d}\frac{\delta_i}{|\delta_i|}\cdot\frac{1}{\sqrt{d}}u_i.
	\end{align}
	We have
	\begin{align*}
		\lVert w_1-{\tilde{w}}\rVert_{{2}}^2= & \sum_{i=1}^{d}\left|\delta_i-\frac{1}{\sqrt{d}}|\delta_i|^{-1}\delta_i\right|^{{2}} \\
		= & \sum_{i=1}^{d}\left(|\delta_i|-\frac{1}{\sqrt{d}}\right)^2 \\
		\leq & {100d^{-1}(d+1)^4\varepsilon^{1/4}}.
	\end{align*}
	Consider the unitary operator $U$ over $\mathbb{C}^d$ given by
	\begin{align*}
		Ux=x-2\frac{w_1-\tilde{w}}{\lVert w_1-\tilde{w}\rVert_{{2}}}\left(\frac{w_1-\tilde{w}}{\lVert w_1-\tilde{w}\rVert_{{2}}}\right)^{\dagger}x.
	\end{align*}
	Then $Uw_1=\tilde{w}$, hence for any $2\leq i\leq d$,
	\begin{align*}
		\lVert w_i-Uw_i\rVert_{{2}}
		= & {2|(w_1-\tilde{w})^{\dagger}(w_i)|} \\
		\leq & 2\lVert w_1-\tilde{w}\rVert_{{2}} \\
		\leq & {20d^{-1/2}(d+1)^2\varepsilon^{1/8}}.
	\end{align*}
	Therefore, for any $i\in[d]$ we have
	\begin{align*}
		& \left\lVert \frac{1}{d+1}(Uw_i)(Uw_i)^{\dagger}-\frac{1}{d+1}w_iw_i^{\dagger}\right\rVert_{\text{op}} \\
		\leq & \frac{1}{d+1}\left(\lVert(Uw_i-w_i)(Uw_i-w_i)^{\dagger}\rVert_{\text{op}}+\lVert(Uw_i-w_i)w_i^{\dagger}\rVert_{\text{op}}+\lVert w_i(Uw_i-w_i)^{\dagger}\rVert_{\text{op}}\right) \\
		\leq & {400d^{-1}(d+1)^3\varepsilon^{1/4}+40d^{-1/2}(d+1)\varepsilon^{1/8}} \\
		\leq & {60d^{-1/2}(d+1)\varepsilon^{1/8}}.
	\end{align*}
	Now consider the correlation $P^{(2)}$ given by the operators $\hat{C}_i$ for $i\in[d]$,
	$\frac{1}{d+1}(Uw_i)(Uw_i)^{\dagger}$ for $i\in[2d]-[d]$, and $D_j$ for $j\in[2d]$.
	For any $i\in[d]$ and $j\in[2d]$, we have
	\begin{align}
		|P^{(2)}_{(d+i),j}-P^{(1)}_{(d+i),j}|= & \left|\mathrm{tr}\left(\left(\frac{1}{d+1}(Uw_i)(Uw_i)^{\dagger}-\frac{1}{d+1}w_iw_i^{\dagger}\right)D_j\right)\right| \\
		\leq & \frac{1}{d+1}\lVert (Uw_i)(Uw_i)^{\dagger}-w_i w_i^{\dagger}\rVert_{\text{op}}\mathrm{tr}(D_j) \\
		\leq & {62d^{-1/2}\varepsilon^{1/8}}.
	\end{align}
	Meanwhile, for any $i\in[d]$ it holds that
	\begin{align*}
		\mathrm{tr}\left(\frac{1}{d+1}(Uw_1)(Uw_1)^{\dagger}\hat{C}_i\right)
		= & \frac{d}{(d+1)^2}|u_i^{\dagger} \tilde{w}|^2 \\
		= & \frac{1}{(d+1)^2}.
	\end{align*}
	In summary we have
	\begin{align}
		D(P^{(2)},P)\leq & D(P^{(2)},P^{(1)})+D(P^{(1)},P') + D(P',P) \\
		\leq & {62d^{-1/2}\varepsilon^{1/8}}+{36\varepsilon^{1/4}} + \varepsilon \\
		\leq & {64d^{-1/2}\varepsilon^{1/8}}.
	\end{align}
	and {that the operators that generate $P^{(2)}$ satisfy the desired five conditions.}
\end{proof}

{Based on the above lemmas, we now adjust the measurement operators $A_i$ and $B_j$ slightly such that they have better mathematical properties, and furthermore, the correlation they generate keeps almost unchanged. }

\begin{lemm}
	\label{lemm:robust-rectify-config}
	Suppose $(\rho,\{A_i\},\{B_j\})$ is a configuration that generates a $2d\times 2d$ correlation $P'_{ij}=\mathrm{tr}(\rho(A_i\otimes B_j))$ with $D(P,P')=\varepsilon$,
	where $P$ is given in Eq.\eqref{correlation:maximal},
	$\varepsilon\leq {(18)^{-9}(d+1)^{-16}}$,
	$\rho$ is the quantum state that Alice and Bob share,
	and $\{A_i\}$ and $\{B_j\}$ are the operators of the POVMs that they perform respectively.
	Then there exist measurement operators $\tilde{A}_i$ and $\tilde{B}_j$
	such that:
	\begin{enumerate}
		\item
			$\forall i\in[2d], \lVert \tilde{A}_i-A_i\rVert_{\text{op}}\leq {128(d+1)\varepsilon^{1/8}},
			\lVert \tilde{B}_i-B_i\rVert_{\text{op}}\leq {128(d+1)\varepsilon^{1/8}}$.
		\item
			$\forall i\in[2d], \mathrm{rank}(\tilde{A}_i)=\mathrm{rank}(\tilde{B}_i)=1$;
		\item
			{$\sum_{i=1}^d\tilde{A}_{i}=\sum_{i=1}^d\tilde{B}_{i}=\frac{d}{d+1}I_d$;}
		\item
            {$\sum_{i=1}^d\tilde{A}_{d+i}=\sum_{i=1}^d\tilde{B}_{d+i}=\frac{1}{d+1}I_d$;}
		\item
			$\forall i\in[d], \mathrm{tr}(\tilde{A}_i \tilde{A}_{d+1})=\mathrm{tr}(\tilde{B}_i \tilde{B}_{d+1})=\frac{1}{(d+1)^2}$;
		\item
			The configuration $(\rho,\{\tilde{A}_i\},\{\tilde{B}_j\})$ produces a correlation $P''_{ij}=\mathrm{tr}(\rho(\tilde{A}_i\otimes\tilde{B}_j))$
			with $D(P,{P''})\leq {256(d+1)\varepsilon^{1/8}}$.
	\end{enumerate}
\end{lemm}
\begin{proof}
	Consider Alice's side first; since $\{A_i\}$ and $\mathrm{tr}((I_A\otimes B_j)\rho)$ form a PSD factorization,
	by Lemma \ref{lemm:robust-rankone} and Lemma \ref{lemm:robust-rectify},
	there are operators $\tilde{A}_i$($i\in[2d]$) that satisfy all their relevant conditions in the statement of the lemma;
	in particular, we have that $\lVert \tilde{A}_i-A_i\rVert_{\text{op}}\leq {128(d+1)\varepsilon^{1/8}}$ and the correlation $P^{(A)}$ given by $\tilde{A}_i$ and $\mathrm{tr}_B((I_d\otimes B_j)\rho)$ is very close to $P'$,
	i.e., $D\left(P', P^{(A)}\right) \leq D\left(P, P^{(A)}\right) + D\left(P', P\right) \leq {128\varepsilon^{1/8}}$.
	Similarly the operators $\tilde{B}_i$ can be constructed within ${128(d+1)\varepsilon^{1/8}}$ operator norm distance from $B_i$.
	It remains to show that the resulting correlation $P''$ produced from $(\rho,\{\tilde{A}_i\},\{\tilde{B}_j\})$ is close to $P$.
	Indeed, it holds that
{
	\begin{align}
			 D\left(P'', P\right)=&\max_{i \in[2 d], j \in[2 d]}\left|\operatorname{tr}\left((\tilde{A_i} \otimes \tilde{B}_j) \rho\right)-\operatorname{tr}\left((A_i \otimes B_j) \rho\right)\right| \\
			 \leq&\max _{i \in[2 d], j \in[2 d]}\left(\left|\operatorname{tr}\left((\tilde{A}_i \otimes \tilde{B}_j) \rho\right)-\operatorname{tr}\left((\tilde{A_i} \otimes B_j) \rho\right)\right|+\left|\operatorname{tr}\left((\tilde{A_i} \otimes B_j) \rho\right)-\operatorname{tr}\left((A_i \otimes B_j) \rho\right)\right|\right) \\
			 \leq &\max _{i \in[2 d], j \in[2 d]}\left|\operatorname{tr}\left((\tilde{A}_i \otimes (\tilde{B}_j-B_j))\rho\right) \right|+{128\varepsilon^{1/8}} \\
			 \leq &\max _{i \in[2 d], j \in[2 d]}\left\|\tilde{A}_i \otimes (\tilde{B}_j-B_j)\right\|_{\text{op}}+{128\varepsilon^{1/8}} \\
			 \leq &\quad \max _{j \in[2d]}\left\|\tilde{B}_j-B_j\right\|_{\text{op}}+{128\varepsilon^{1/8}} \\
			 \leq & {256(d+1)\varepsilon^{1/8}}.
			\end{align}	
}
\end{proof}

{The lemma below shows that the desirable properties of the new measurement operators $\tilde{A}_i$ and $\tilde{B}_j$ allow us to obtain a lower bound on some distance between the target quantum state and a $d\times d$ maximally entangled pure state, which is achieved by looking into an optimization problem and its dual problem.}
\begin{lemm}
	\label{lemm:robust-sdp}
	Suppose $(\rho,\{A_i\},\{B_j\})$ forms a correlation $P'$ such that $D(P',P)\leq \varepsilon$
	and the measurement operators $A_i$ and $B_j$ satisfy the first four conditions specified in the
	statement of Lemma \ref{lemm:robust-rectify-config}.
	Then
	\begin{align}
		\mathrm{tr}(\tilde{\rho}\rho)\geq 1-22d^3\varepsilon,
	\end{align}
	where $\tilde{\rho}$ is a $d\times d$ maximally entangled pure state.
\end{lemm}
\begin{proof}
	Due to the conditions satisfied by the measurement operators $A_i$ and $B_j$,
	we may apply local unitaries such that for $i\in[d]$ it holds that $A_i=B_i=\frac{d}{d+1}e_ie_i^{\dagger}$
	and $A_{d+1}=B_{d+1}=\frac{1}{{d(d+1)}}\sum_{ij}e_i e_j^{\dagger}$;
	henceforth we assume that these operators are indeed given as such.
	Define $\tilde{\rho}=\frac{1}{d}\sum_{ij}(e_i\otimes e_i)(e_j\otimes e_j)^{\dagger}$;
	then $\tilde{\rho}$ is a maximally entangled pure state.
	Consider the optimization problem with variable $\hat{\rho}$,
	\begin{equation}\nonumber
	    \begin{aligned}
		\text{minimize } & \mathrm{tr}(\tilde{\rho}\hat{\rho}), \\
		\text{subject to } & \mathrm{tr}(\hat{\rho}(A_i\otimes B_j)=P'_{ij}) \ \ \forall i,j\in[d+1]\\
		& {\hat{\rho}}\geq 0
	\end{aligned}.
	\end{equation}
	Notice that the quantum state $\rho$ in the configuration is a feasible solution to the above optimization problem. One can verify that the dual problem to the above optimization problem is
	\begin{equation}\nonumber
	    \begin{aligned}
		\text{maximize } & \sum_{i,j=1}^{d+1}\mu_{ij}P'_{ij} \\
		\text{subject to } & \tilde{\rho}-\sum_{i,j=1}^{d+1}\mu_{ij}(A_i\otimes B_j)\geq 0
	\end{aligned},
	\end{equation}
	where the dual variables are $\mu_{ij}\in\mathbb{R}$.
	By straightforward calculation, one can show that a feasible solution $\mu$ to this dual problem can be constructed as below,
	\begin{align*}
		\mu_{ij}=
		\begin{cases}
			\frac{(d+1)^2}{d^2}, & i=j\in[d] \\
			-\frac{(d+1)^2}{d^2}\left[\left(\frac{3d-4}{d(d+1)^2}K\right)\frac{4(d-1)(d-2)K}{2dK-d(d+1)^2}+\frac{4(d-2)^2}{d(d+1)^2}K\right], & i,j\in[d], i\neq j \\
			-K, & (i=d+1\land j\neq d+1)\lor(i\neq d+1\land j=d+1) \\
			2dK, & i=j=d+1
		\end{cases},
	\end{align*}
	where $K\equiv (d+1)^2$. {According to the weak duality theorem, this gives a lower bound for the output of the primal problem, i.e.,
	\begin{align}
	\mathrm{tr}(\tilde{\rho}\hat{\rho})\geqslant & 1-\left(\frac{d+1}{d}\right)^2 \cdot d \varepsilon-\frac{(d+1)^2(d-1)}{d}\left[\frac{4(3 d-4)(d-1)(d-2)}{d^2}+\frac{4(d-2)^2}{d}\right] \varepsilon-2 d(d+1)^2\left(\frac{1}{d(d+1)^2}+\varepsilon\right) \\
	& +2 d(d+1)^2\left(\frac{1}{d(d+1)^2}-\varepsilon\right) \\
	\geqslant & 1-\varepsilon\left(d+2+\frac{1}{d}+4 d(d+1)^2+\frac{4(d+1)^2(d-1)^2(3 d-4)(d-2)}{d^3}+\frac{4(d+1)^2(d-2)^2(d-1)}{d^2}\right) \\
	\geqslant & 1-\varepsilon\left(d+3+4 d(d+1)^2+\frac{16(d+1)^2(d-1)^2}{d}\right) \\
	\geqslant & 1-\varepsilon\left(20 d^3-24 d^2+5 d+11\right) \\
	\geqslant & 1-22d^3\varepsilon.
	\end{align}
	}
\end{proof}

Putting all the above results together, we eventually obtain the following theorem.
\begin{theo}\label{thm:robustness}
If Alice and Bob generate a $2d\times 2d$ correlation $P'$ by locally measuring a $d\times d$ bipartite quantum state $\rho$, such that $|P_{xy}-P'_{xy}|\leq\epsilon$ for any $x,y\in[2d]$, where $P$ is the correlation in Eq.\eqref{correlation:maximal} and $\epsilon\leq {(18)^{-9}(d+1)^{-16}}$, then the fidelity between $\rho$ and a maximally entangled pure state will be at least ${1-5632(d+1)^4\varepsilon^{1/8}}$.
\end{theo}

{We stress that in the proof of the above theorem, in addition to the quantum dimension $d$, all the mathematical structures are obtained based on the observed correlation $P'$ only, which means that the quantum state certification is conducted without any other assumptions on the underlying quantum state $\rho$ (i.e., it is pure or not) and the quality of the quantum measurements $\{A_i\}$ and $\{B_j\}$. Therefore, the protocol is a semi-device-independent one, though no quantum nonlocality is involved.}

	\subsection{Proof for the Robustness under Correlation Scaling}

	We first reproduce the scaling technique with an alternative choice of parameters. 
	\begin{theo}
		Suppose $\lambda_i>0$ for $i\in[d]$.
		For any real number {$\alpha_0\in(0,\sqrt{d}(d+1)\min_{i}\sqrt{\lambda_i})$},
		let
		\begin{align}
			\alpha_i= & \frac{d+1}{\sqrt{d}}\sqrt{\lambda_i}-\frac{1}{d}\alpha_0
		\end{align}
		for $i\in[d]$.
		Let
		\begin{align}
			\vec{\alpha}= & (\alpha_1,\cdots,\alpha_d,\alpha_0,\cdots,\alpha_0).
		\end{align}
		If the correlation
		\begin{align}
			{\hat{P}}= & \mathrm{diag}(\vec{\alpha})P\mathrm{diag}(\vec{\alpha})
		\end{align}
		is obtained upon locally measuring a $d\times d$ state $\rho$, then $\rho$ is bipartite pure entangled state
		with Schmidt coefficients $\sqrt{\lambda_1},\cdots,\sqrt{\lambda_d}$.
	\end{theo}

	\begin{proof}
		The exact same proof for the scaling technique for the $2d\times 2d$ construction applies,
		except that the spectrum of the normal form is now equal to that of
		\begin{align}
			\left(\frac{\sqrt{d}}{d+1}\sum_{i=1}^{d}\alpha_i e_i e_i^{\dagger}+\frac{1}{(d+1)\sqrt{d}}\alpha_0 I_d\right)^2,
		\end{align}
		which by the constructions of $\alpha_0$ and $\alpha_i$($i\in[d]$) is exactly $\lambda_1,\cdots,\lambda_d$.
	\end{proof}

	Suppose $(\rho,\{A_i\},\{B_j\})$ is a configuration that gives a correlation $P'$ close to the scaled $2d\times 2d$ correlation
	\begin{align*}
		\hat{P}\equiv \mathrm{diag}(\vec{\alpha})P\mathrm{diag}(\vec{\alpha});
	\end{align*}
	that is, $D(P',\hat{P})\leq\varepsilon_0$.
	Then
	\begin{align*}
		(\rho,\{\vec{\alpha}_1^{-1}A_1,\cdots,\vec{\alpha}_{2d}^{-1}A_{2d}\},\{\vec{\alpha}_1^{-1}B_1,\cdots,\vec{\alpha}_{2d}^{-1}B_{2d}\})
	\end{align*}
	is a configuration that realizes a correlation $P''$ such that {$D(P'',{P})\leq\max_{i\in[2d]}(\alpha_i)^{-2}\varepsilon_0\equiv\varepsilon_1$}.
	Notice that this configuration is not physical since the measurement operators on each site do not sum to $I_d$.
	Thus we define $S_A=\sum_{i=1}^{2d}\vec{\alpha}_{i}^{-1}A_i$ and $S_B=\sum_{j=1}^{2d}\vec{\alpha}_{j}^{-1}B_j$,
	and consequently
	\begin{align*}
		((S_A^{1/2}\otimes S_B^{1/2})\rho(S_A^{1/2}\otimes S_B^{1/2}),\{S_A^{-1/2}\vec{\alpha}_i^{-1}A_i S_A^{-1/2}\},\{S_B^{-1/2}\vec{\alpha}_j^{-1} B_j S_B^{-1/2}\})
	\end{align*}
	is a physical configuration that realizes $P''$.
	By the previous lemmas, there exists a state $\tilde{\rho}$ and measurement operators $\tilde{A}_i$ and $\tilde{B}_j$ such that:
	\begin{itemize}
		\item
			$(\tilde{\rho},\{\tilde{A}_i\},\{\tilde{B}_j\})$ realizes the original correlation $P$;
		\item
			For $i\in[d]$, it holds that
			\begin{align}
				\lVert \tilde{A}_i-S_A^{-1/2}\vec{\alpha}_i^{-1}A_i S_A^{-1/2}\rVert_{\text{op}}\leq & f(\varepsilon_1), \\
				\lVert \tilde{B}_j-S_B^{-1/2}\vec{\alpha}_j^{-1}B_j S_B^{-1/2}\rVert_{\text{op}}\leq & f(\varepsilon_1),
			\end{align}
			 {where $f(\varepsilon_1)$ can be determined by the lemmas in the previous subsection}.
		\item
			\begin{align}
				\left\lVert \sum_{i=1}^{d}\tilde{A}_{d+i}-S_A^{-1/2}\left(\sum_{i=1}^{d}\alpha_0^{-1}A_{d+i}\right)S_A^{-1/2}\right\rVert_{\text{op}}\leq & g(\varepsilon_1), \\
				\left\lVert \sum_{j=1}^{d}\tilde{B}_{d+j}-S_B^{-1/2}\left(\sum_{j=1}^{d}\alpha_0^{-1}B_{d+j}\right)S_B^{-1/2}\right\rVert_{\text{op}}\leq & g(\varepsilon_1),
			\end{align}
			 {where $g(\varepsilon_1)$ can be determined by the lemmas in the previous subsection}.
		\item
			\begin{align}
				\left\lVert\tilde{\rho}-(S_A^{1/2}\otimes S_B^{1/2})\rho(S_A^{1/2}\otimes S_B^{1/2})\right\rVert_{\text{op}}\leq h(\varepsilon_1),
			\end{align}
   {where $h(\varepsilon_1)$ can be determined by the lemmas in the previous subsection}.
	\end{itemize}
	Now define $\tilde{S}_A=\sum_{i}\vec{\alpha}_i\tilde{A}_i$ and $\tilde{S}_B=\sum_{j}\vec{\alpha}_j\tilde{B}_j$.
	Then the configuration
	\begin{align*}
		((\tilde{S}_A^{1/2}\otimes\tilde{S}_B^{1/2})\tilde{\rho}(\tilde{S}_A^{1/2}\otimes\tilde{S}_B^{1/2}),\{\tilde{S}_A^{-1/2}\vec{\alpha}_i\tilde{A}_i\tilde{S}_A^{-1/2}\},\{\tilde{S}_B^{-1/2}\vec{\alpha}_j\tilde{B}_j\tilde{S}_B^{-1/2}\})
	\end{align*}
	exactly realizes {the scaled correlation $\hat{P}$}, hence $\hat{\rho}\equiv(\tilde{S}_A^{1/2}\otimes\tilde{S}_B^{1/2})\tilde{\rho}(\tilde{S}_A^{1/2}\otimes\tilde{S}_B^{1/2})$
	is a pure state with Schmidt coefficients determined from the construction of the scaling factors.
	We now estimate $\lVert\rho-\hat{\rho}\rVert_{\text{op}}$:
    \begin{align*}
    	& \lVert\rho-\hat{\rho}\rVert_{\text{op}} \\
    	= & \left\lVert\rho-{(\tilde{S}_A^{1/2}\otimes\tilde{S}_B^{1/2})\tilde{\rho}(\tilde{S}_A^{1/2}\otimes\tilde{S}_B^{1/2})}\right\rVert_{\text{op}} \\
    	\leq & \left\lVert\rho-(\tilde{S}_A^{1/2}\otimes\tilde{S}_B^{1/2})(S_A^{1/2}\otimes S_B^{1/2})\rho(S_A^{1/2}\otimes S_B^{1/2})(\tilde{S}_A^{1/2}\otimes\tilde{S}_B^{1/2})\right\rVert_{\text{op}} \\
    	& +\left\lVert(\tilde{S}_A^{1/2}\otimes\tilde{S}_B^{1/2})(S_A^{1/2}\otimes S_B^{1/2})\rho(S_A^{1/2}\otimes S_B^{1/2})(\tilde{S}_A^{1/2}\otimes\tilde{S}_B^{1/2}) -(\tilde{S}_A^{1/2}\otimes\tilde{S}_B^{1/2})\tilde{\rho}(\tilde{S}_A^{1/2}\otimes\tilde{S}_B^{1/2})\right\rVert_{\text{op}} \\
    	\leq & \left\lVert\rho-(\tilde{S}_A^{1/2}\otimes\tilde{S}_B^{1/2})(S_A^{1/2}\otimes S_B^{1/2})\rho\right\rVert_{\text{op}} \\
    	& +\left\lVert(\tilde{S}_A^{1/2}\otimes\tilde{S}_B^{1/2})(S_A^{1/2}\otimes S_B^{1/2})\rho-(\tilde{S}_A^{1/2}\otimes\tilde{S}_B^{1/2})(S_A^{1/2}\otimes S_B^{1/2})\rho(S_A^{1/2}\otimes S_B^{1/2})(\tilde{S}_A^{1/2}\otimes\tilde{S}_B^{1/2})\right\rVert_{\text{op}} \\
    	& +\lVert \tilde{S}_A\rVert_{\text{op}}\lVert\tilde{S}_B\rVert_{\text{op}}h(\varepsilon_1) \\
    	\leq & \lVert I-(\tilde{S}_A^{1/2}\otimes \tilde{S}_{B}^{1/2})(S_A^{1/2}\otimes S_B^{1/2})\rVert_{\text{op}} \\
    	& +\lVert(\tilde{S}_A^{1/2}\otimes \tilde{S}_{B}^{1/2})(S_A^{1/2}\otimes S_B^{1/2})\rVert_{\text{op}}
    	\lVert I-(S_A^{1/2}\otimes S_B^{1/2})(\tilde{S}_A^{1/2}\otimes \tilde{S}_{B}^{1/2})\rVert_{\text{op}} \\
    	& +\lVert \tilde{S}_A\rVert_{\text{op}}\lVert\tilde{S}_B\rVert_{\text{op}}h(\varepsilon_1) \\
    	\leq & \lVert S_A^{-1/2}\otimes S_B^{-1/2}-\tilde{S}_A^{1/2}\otimes \tilde{S}_B^{1/2}\rVert_{\text{op}}
    	\left(\lVert S_A\rVert_{\text{op}}^{1/2}\lVert S_B\rVert_{\text{op}}^{1/2}
    	+\lVert \tilde{S}_A\rVert_{\text{op}}^{1/2}\lVert \tilde{S}_B\rVert_{\text{op}}^{1/2}
    	\lVert S_A\rVert_{\text{op}}\lVert S_B\rVert_{\text{op}}\right) \\
    	& +\lVert \tilde{S}_A\rVert_{\text{op}}\lVert\tilde{S}_B\rVert_{\text{op}}h(\varepsilon_1) \\
    	\leq & {\left({\lVert S_A^{-1}\rVert_{\text{op}}^{1/2}}\lVert S_B^{-1}-\tilde{S}_B\rVert_{\text{op}}^{1/2}+\lVert \tilde{S}_B\rVert_{\text{op}}^{1/2}\lVert S_A^{-1}-\tilde{S}_A\rVert_{\text{op}}^{1/2}\right)
    	\left(\lVert S_A\rVert_{\text{op}}^{1/2}\lVert S_B\rVert_{\text{op}}^{1/2}
    	+\lVert \tilde{S}_A\rVert_{\text{op}}^{1/2}\lVert \tilde{S}_B\rVert_{\text{op}}^{1/2}
    	\lVert S_A\rVert_{\text{op}}\lVert S_B\rVert_{\text{op}}\right)} \\
    	& {+\lVert \tilde{S}_A\rVert_{\text{op}}\lVert\tilde{S}_B\rVert_{\text{op}}h(\varepsilon_1)}.
    \end{align*}
	{To estimate the above upper bound for $\lVert\rho-\hat{\rho}\rVert_{\text{op}}$, note that}
	\begin{align*}
		S_A= & \sum_{i}\vec{\alpha}_i^{-1}A_i\leq \sum_{i}\left(\max_{i'}\vec{\alpha}_{i'}^{-1}\right)A_i=\left(\max_{i'}\vec{\alpha}_{i'}^{-1}\right)I_d,
	\end{align*}
	hence
	\begin{align}
		{\lVert S_A\rVert_{\text{op}}\leq\max_{i'}\vec{\alpha}_{i'}^{-1}}.
	\end{align}
	Similarly, it holds that
	\begin{align}
		{\lVert S_A^{-1}\rVert_{\text{op}}\leq\max_{i'}\vec{\alpha}_{{i'}}},
		\lVert S_B\rVert_{\text{op}}\leq\max_{i'}\vec{\alpha}_{i'}^{-1},
		\lVert \tilde{S}_A\rVert_{\text{op}}\leq\max_{i'}\vec{\alpha}_{{i'}},
		\lVert \tilde{S}_B\rVert_{\text{op}}\leq\max_{i'}\vec{\alpha}_{{i'}}.
	\end{align}
	To estimate $\lVert S_A^{-1}-\tilde{S}_A\rVert_{\text{op}}$, note that
	\begin{align*}
		& \left\lVert S_A^{-1}-\tilde{S}_A\right\rVert_{\text{op}} \\
		= & \left\lVert S_A^{-1}-\sum_{i=1}^{d}\alpha_i\tilde{A}_i-\alpha_0\sum_{i=1}^{d}\tilde{A}_{d+i}\right\rVert_{\text{op}} \\
		= & \bigg\lVert S_A^{-1}-\sum_{i=1}^{d}\alpha_i(\alpha_i^{-1}S_A^{-1/2}A_i S_A^{-1/2}+(\tilde{A}_i-\alpha_i^{-1}S_A^{-1/2}A_i S_A^{-1/2})) \\
		& -\alpha_0\left(\alpha_0^{-1}S_A^{-1/2}\sum_{i=1}^{d}A_{d+i} S_A^{-1/2}+\sum_{i=1}^{d}\tilde{A}_{d+i}-\alpha_0^{-1}S_A^{-1/2}\sum_{i=1}^{d}A_{d+i} S_A^{-1/2}\right)\bigg\rVert_{\text{op}} \\
		\leq & \sum_{i=1}^{d}\alpha_i\lVert \tilde{A}_i-\alpha_i^{-1}S_A^{-1/2}A_i S_A^{-1/2}\rVert_{\text{op}}+\alpha_0 \left\lVert\sum_{i=1}^{d}\tilde{A}_{d+i}-\alpha_0^{-1}S_A^{-1/2}\left(\sum_{i=1}^{d}A_{d+i}\right)S_A^{-1/2}\right\rVert_{\text{op}} \\
		\leq & \left(\sum_{i=1}^{d}\alpha_i\right)f(\varepsilon_1)+\alpha_0 g(\varepsilon_1).
	\end{align*}
	Similarly,
	\begin{align*}
		\lVert S_B^{-1}-\tilde{S}_B\rVert_{\text{op}}\leq \left(\sum_{i=1}^{d}\alpha_i\right)f(\varepsilon_1)+\alpha_0 g(\varepsilon_1).
	\end{align*}
	Therefore, by denoting
	\begin{align*}
		\alpha_{+}= & \max_{i}\vec{\alpha}_i, \\
		\alpha_{-}= & \min_{i}\vec{\alpha}_i, \\
	\end{align*}
	we obtain that
	\begin{align}
		{\lVert \rho-\hat{\rho}\rVert_{\text{op}}\leq 2\alpha_{+}^{1/2}(\alpha_{-}^{{-1}}+\alpha_{+}\alpha_{-}^{{-2}})\left(\left(\sum_{i=1}^{d}\alpha_i\right)f(\varepsilon_1)+\alpha_0 g(\varepsilon_1)\right)^{1/2}+\alpha_{+}^2 h(\varepsilon_1).}
	\end{align}

	Notice that the robustness highly depends on the scaling factors $\alpha_i$($i\in[d]$) and $\alpha_0$; if one of them is extremely close to $0$ or extremely large, then the robustness requirement becomes extremely stringent accordingly.

\subsection{Numerical evidence for the robustness of our protocols}
In this subsection, we provide numerical evidence showing the robustness of our approach. We would like to point out that compared with the theoretical result given in Theorem \ref{thm:robustness}, the robustness that the numerical evidence suggests is much stronger.

The numerical approach we adopt works as follows. Suppose Alice and Bob measure a $d\times d$ quantum state $\rho$ with the POVMs $\{A_i\}$ and $\{B_j\}$ respectively, and obtain a correlation $P'$ that is close to $P$, where $P$ is given in Eq.\eqref{correlation:maximal}. We hope to estimate the fidelity between $\rho$ and $d\times d$ maximally entangled states. For this, we consider the following optimization problem
\begin{equation}\label{eq:optimization_origin}
    \begin{aligned}
        \min_{\substack{\rho, A_i, B_j \geq 0 \\ i=1, \ldots, 2d \\ j=1, \ldots, 2d}}  &\sum_{i=1}^{2d} \sum_{j=1}^{2d}\left(P_{ij}-\mathrm{tr}(\rho A_i \otimes B_j) \right)^2 \\
    \text{subject } \text{to  }&   \sum_{i=1}^{2d} A_i=\sum_{j=1}^{2d} B_j = I_d\\
	&\mathrm{tr}(\rho) = 1
    \end{aligned},
\end{equation}
{where the variables are the quantum state $\rho$ and two sets of POVM operators $\{A_i\}$ and $\{B_j\}$.}
	%
This optimization problem is non-convex and hard to solve directly, and we design a see-saw algorithm to address it, where we transfer the problem in Eq.\eqref{eq:optimization_origin} into several convex optimization problems that are relatively easier to handle by fixing parts of the variables by turns. The pseudo-code for this algorithm is detailed in Algorithm \ref{alg_seesaw}.


\begin{algorithm}
\renewcommand{\algorithmicrequire}{\textbf{Input:}}
\renewcommand{\algorithmicensure}{\textbf{Output:}}	
\caption{{Finding $\rho$ by alternating optimizations}}
	\hspace{-80mm}\textbf{INPUT:} $P \in \mathbb{R}_{+}^{2d \times 2d}$, $\{A_i\}$, $\{B_j\}$, and the stopping precision $\epsilon_2$.\\
	\begin{algorithmic}[1]
		\WHILE{$\sum_{i=1}^{2d} \sum_{j=1}^{2d}\left(P_{ij}-\mathrm{tr}(\rho A_i \otimes B_j) \right)^2 > \epsilon_2$}
		    \STATE  $\qquad \rho \leftarrow$ solve optimization problem (\ref{eq:optimization_origin}) with fixed $\{A_i\}$ and $\{B_j\}$
			\STATE $\qquad\{A_i\} \leftarrow$ solve optimization problem (\ref{eq:optimization_origin}) with fixed $\rho$ and $\{B_j\}$
			\STATE $\qquad\{B_j\} \leftarrow$ solve optimization problem (\ref{eq:optimization_origin}) with fixed $\rho$ and $\{A_i\}$
		\ENDWHILE
	\end{algorithmic}
        \hspace{-150mm}\textbf{OUTPUT:} $\rho$.
	\label{alg_seesaw}
\end{algorithm}
Here we choose $P$ as in Eq.\eqref{correlation:maximal}, and each matrix in the initial inputs $\{A_i\}$ and $\{B_j\}$ is given by
\begin{equation}\label{eq:initial-points}
    \sum_{k=1}^d a^k a^{k^T},
\end{equation}
    where for each $k\in[d]$, $a^k$ is an $n$-dimensional vector whose entries are initialized according to the normal distribution $\mathcal{N}(0,1)$. For a given stopping precision $\epsilon_2$, we solve the convex optimization problems in Algorithm \ref{alg_seesaw} by CVXPY \cite{diamond2016cvxpy}. After Algorithm \ref{alg_seesaw} returns a quantum state $\rho$, we calculate the largest fidelity between it and a maximally entangled quantum state:
\begin{align*}
    \max_{U_A,U_B} \mathrm{tr}((U_A\otimes U_B) \rho(U_A^\dagger\otimes U_B^\dagger)\ket{\psi}\bra{\psi}),
    \end{align*}
where $U_A$ and $U_B$ are local unitary matrices that are performed on Alice and Bob's sides respectively. In fact, this optimization problem is also hard to solve exactly, so we instead compute a lower bound for its solution. Specifically, after obtaining $\rho$, we first compute its spectral decomposition, that is, $\rho = \sum_{i=1}^{d^2} x_i\left|i_{A B}\right\rangle\left\langle i_{A B}\right|$, where $x_1 \geqslant x_2 \geq \cdots \geqslant x_{d^2}$. It can be seen that
$$\lVert \rho-\ket{1_{A B}}\bra{1_{A B}}\rVert_{\text{op}} \leq \max \left\{\left(1-x_1\right), x_2\right\} {= 1-x_1}. $$
Then it holds that
\begin{align*}
	& \operatorname{tr}\left(U_A \otimes U_B \rho U_A^{\dagger} \otimes U_B^{\dagger}\ket{\psi} \bra{\psi}\right) \\
	& \geqslant \operatorname{tr}\left(U_A \otimes U_B\left(\ket{1_{A B}}\bra{1_{A B}} -{(1-x_1)}I_d\right)U_A^{\dagger} \otimes U_B^{\dagger}\ket{\psi} \bra{\psi}\right)\\
	& = \operatorname{tr}\left(U_A \otimes U_B\ket{1_{A B}}\bra{ 1_{A B}}U_A^{\dagger} \otimes U_B^{\dagger}\ket{\psi} \bra{\psi}\right){-1+x_1}.
	\end{align*}
Suppose the Schmidt coefficients of $\ket{1_{A B}}$ are $\sqrt{\lambda_1},\cdots,\sqrt{\lambda_d}$, then we have that
\begin{align}\label{lower bound:fidelity}
    \max_{U_A,U_B} \mathrm{tr}((U_A\otimes U_B) \rho (U_A^\dagger\otimes U_B^\dagger)\ket{\psi}\bra{\psi})\geq\left(\sum_i \sqrt{\lambda_i}\right)^2 / d-{1+x_1}.
    \end{align}
Here we choose $\epsilon_2\in[10^{-6},10^{-4}]$ at intervals of size $10^{-6}$, and on each value of $\epsilon_2$ we run Algorithm \ref{alg_seesaw} for 2,000 times. Note that due to the way we pick up $\{A_i\}$ and $\{B_j\}$ in Eq.(\ref{eq:initial-points}), each running of our algorithm has different inputs. Based on the output $\rho$ of Algorithm \ref{alg_seesaw}, we lower bound its largest fidelity with a maximally entangled state with Eq.(\ref{lower bound:fidelity}). For $d={3}$, the results are depicted in the figure below.
\begin{figure}[!h]
	\centering
	\includegraphics[width=0.5\textwidth]{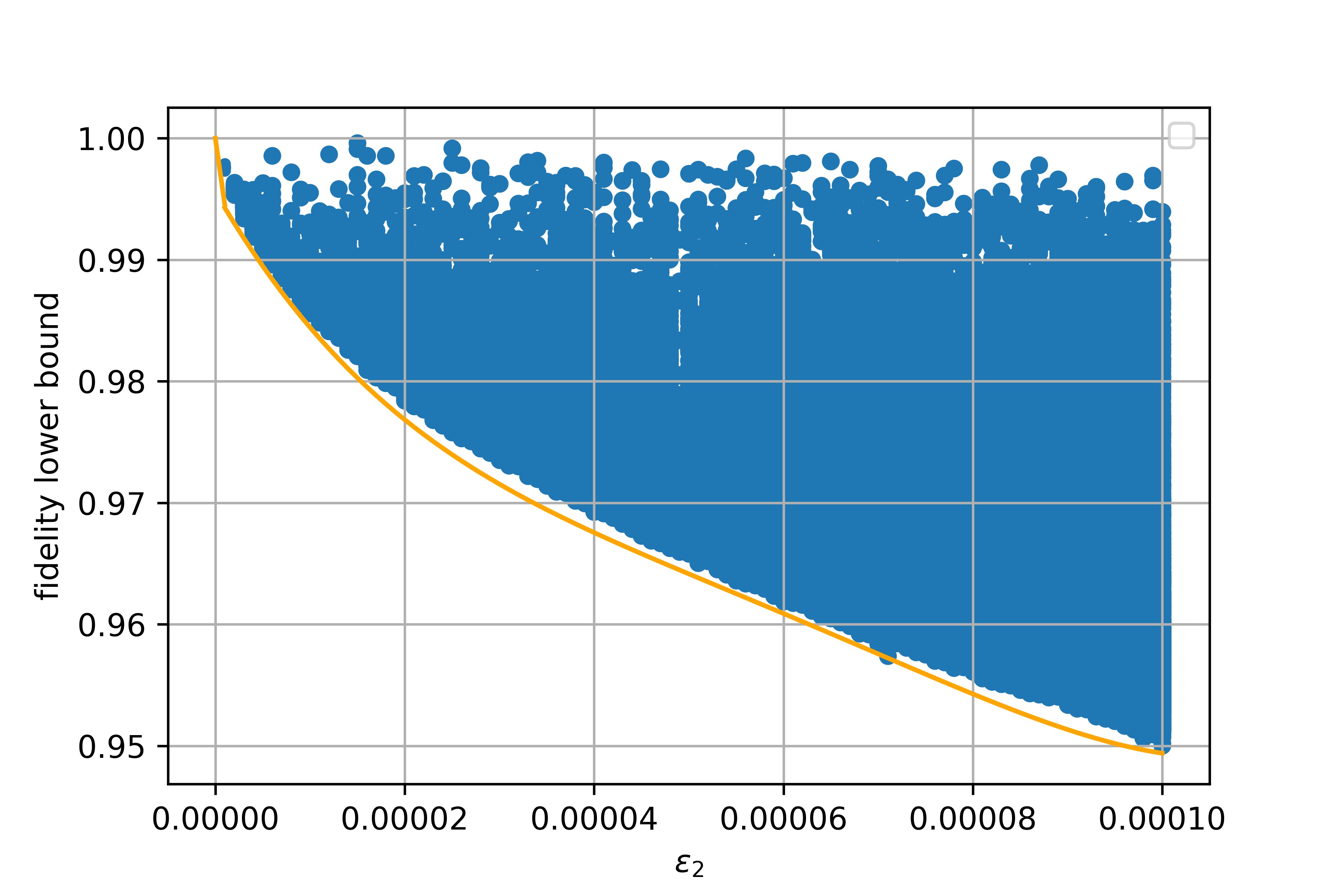}
	\caption{The lower bound for the largest fidelity between the output quantum states of Algorithm \ref{alg_seesaw} and a maximally entangled state, where $d={3}$ and the stopping precision ranges from $10^{-6}$ to $10^{-4}$.}
	\label{fig:seesaw_2}
\end{figure}

Note that when $D(P,P')\approx {1.67}\times 10^{-3}$, the corresponding $\epsilon_2$ is roughly $10^{-4}$, and Fig.\ref{fig:seesaw_2} shows that the largest fidelity between any quantum state $\rho$ that can produce $P'$ and a maximally entangled state is lower bounded by approximately {0.95}. Considering that the average value of $P$'s entries is {around $2.78\times 10^{-2}$}, this conclusion suggests that our protocol has a very decent robustness.

\section{Certifying bipartite pure states with $(d+1)\times(d+1)$ correlations}

    We first show that the maximally entangled pure state can be certified with $(d+1)\times(d+1)$ correlations. Consider the correlation $R$ of size $(d+1)\times(d+1)$ given by
	\begin{align}\label{correlation:maximal2}
		R(i,j)= &
		\begin{cases}
			{\frac{d}{(d+1)^2}}, & i=j \\
			{\frac{1}{d(d+1)^2}}, & i\neq j
		\end{cases}.
	\end{align}
Then we have the following result.
\begin{theo}\label{thm:maximal2}
If Alice and Bob generate the correlation $R$ in Eq.\eqref{correlation:maximal2} by locally measuring a $d\times d$ bipartite quantum state $\rho$, then $\rho$ must be a maximally entangled pure state.
\end{theo}

	Let $\{C_i\}$ and $\{D_j\}$ be a canonical form of PSD decomposition for $R$;
	that is, $\sum_i C_i=\sum_jD_j=\Lambda$ where $\Lambda$ is a diagonal matrix.
	By summing up all the entries of $R$, we obtain the following proposition:
	\begin{prop}
		$\mathrm{tr}({\Lambda^2})={1}$.
	\end{prop}

	The following lemma utilizes the dimensionality constraint.
	\begin{lemm}
		Let $X$ be a nonzero $d\times d$ positive semi-definite matrix.
		Then
		\begin{align}
			\sqrt{\mathrm{tr}(X^2)}\leq \mathrm{tr}(X)\leq \sqrt{d\cdot\mathrm{tr}(X^2)},
		\end{align}
		where the left equality is attained only when $X$ has rank $1$ and the right equality is attained only when $X$ is a scalar multiple of the identity.
	\end{lemm}
	\begin{proof}
		The inequality on the left can be shown by representing the quantities in terms of eigenvalues.
		The inequality on the right can be shown via Cauchy-Schwarz inequality.
	\end{proof}

	We now have
	\begin{align*}
		{2\sqrt{d}}= & 2\sum_i\sqrt{\mathrm{tr}(C_i D_i)} \\
		\leq & 2\sum_i\left(\mathrm{tr}(C_i^2)\mathrm{tr}(D_i^2)\right)^{1/4} \\
		\leq & \sum_i\sqrt{\mathrm{tr}(C_i^2)}+\sum_i\sqrt{\mathrm{tr}(D_i^2)} \\
		\leq & \sum_i\mathrm{tr}(C_i)+\sum_i\mathrm{tr}(D_i) \\
		= & 2\mathrm{tr}({\Lambda}) \\
		\leq & 2\sqrt{d\mathrm{tr}({\Lambda^2})} \\
		= & {2\sqrt{d}},
	\end{align*}
	hence all the involved inequalities are tight.
	This implies that ${\Lambda}$ is a scalar multiple of the identity. Furthermore, all matrices $C_i$ and $D_i$ are of rank $1$, {as for any $i$ it holds that $\sqrt{\mathrm{tr}(C_i^2)}=\mathrm{tr}(C_i)$ and $\sqrt{\mathrm{tr}(D_i^2)}=\mathrm{tr}(D_i)$}.
	Additionally, from
	\begin{align*}
		\mathrm{tr}(C_i D_i)=\sqrt{\mathrm{tr}(C_i^2)\mathrm{tr}(D_i^2)}
	\end{align*}
	we conclude that $C_i$ is a scalar multiple of $D_i$.
	Thus, by $\mathrm{tr}(C_i D_i)={\frac{d}{(d+1)^2}}$ we have $\mathrm{tr}(C_i)+\mathrm{tr}(D_i)\geq {2\frac{\sqrt{d}}{d+1}}$ with equality attained
	only when $\mathrm{tr}(C_i)=\mathrm{tr}(D_i)={\frac{\sqrt{d}}{d+1}}$.
	By $\sum_i\mathrm{tr}(C_i)+\sum_i\mathrm{tr}(D_i)=2\sqrt{d}$ we finally have $\mathrm{tr}(C_i)=\mathrm{tr}(D_i)={\frac{\sqrt{d}}{d+1}}$,
	which implies that $C_i=D_i$ for all $i\in[d+1]$.

    {Suppose Alice and Bob generate $R$ by measuring $\rho$ with local quantum measurements $\{A_i\}$ and $\{B_j\}$. We now investigate the measurement operators $A_i$ and $B_j$}.

	\begin{lemm}
		All measurement operators $A_i$ and $B_j$ have rank 1.
	\end{lemm}
	\begin{proof}
		We prove the lemma for $A_i$, {and the case for $B_j$ is similar}.
		Notice that $\{A_i\}$ and {$\{\mathrm{tr}_B((I_A\otimes B_j)\rho)\}$} form a PSD factorization for $R$.
		This PSD factorization is equivalent to a canonical form of PSD factorization for $R$, where the ranks of the respective operators are equal. Thus by the previous discussion we thus have $\mathrm{rank}(A_i)=1$ for all $i\in[d+1]$.
	\end{proof}

	\begin{lemm}
		$\mathrm{tr}(A_i)=\frac{d}{d+1}$.
	\end{lemm}
	\begin{proof}
		We have proved that for an arbitrary canonical form of PSD decomposition for $R$, ${\Lambda}$ is a scalar multiple of the identity, then we have $\mathrm{tr}_B(\rho)={I_d/d}$.
		Therefore, by summing up the corresponding row of $R$ {and calculating the marginal probability of Alice}, we have that
		\begin{align*}
			\mathrm{tr}(A_i)= & d\cdot\mathrm{tr}(A_i\mathrm{tr}_B(\rho)) \\
			= & d\left(\frac{d}{(d+1)^2}+\frac{1}{d(d+1)^2}\cdot d\right) \\
			= & \frac{d}{d+1}.
		\end{align*}
	\end{proof}

	Let $\tilde{B}_i\equiv \mathrm{tr}_B((I_A\otimes B_i)\rho)$. We have the following conclusion.
	\begin{lemm}
		\label{lemm:mirror}
		$\tilde{B}_i=\frac{1}{d} A_i$
		for all $i\in[d+1]$.
	\end{lemm}
	\begin{proof}
		According to Lemma \ref{correlation:scsd} and the above discussion, for any PSD factorization $\{C_i\}$ and $\{D_i\}$ for $R$ of size $d\times d$, we have $S_C^{1/2}S_D S_C^{1/2}={\frac{1}{d}} I$,
		where $S_C=\sum_i C_i$ and $S_D=\sum_i D_i$.
		Since $\sum_i A_i=I$,
		we must have $\sum_i \tilde{B}_i={\frac{1}{d}} I$.
		Therefore, $\{{\frac{1}{\sqrt{d}}}A_i\}$ and $\{\sqrt{d}\tilde{B}_i\}$ is a canonical form of PSD factorization for $R$.
		From the above discussion we then have ${\frac{1}{\sqrt{d}}}A_i={\sqrt{d}}\tilde{B}_i$.
	\end{proof}

	We single out the operator $A_{d+1}$ to obtain
	\begin{align}
		\sum_{i=1}^{d}(I-A_{d+1})^{-1/2}A_i(I-A_{d+1})^{-1/2}=I.
	\end{align}
	Since this is a resolution of the identity involving $d$ rank one positive operators,
	these positive operators must be orthogonal projections on some orthonormal basis.
	Thus, upon suitable local unitary, we may assume that
	\begin{align}
		(I-A_{d+1})^{-1/2}A_i(I-A_{d+1})^{-1/2}=e_i e_i^{\dagger},
	\end{align}
	where $\{e_i\}$ is an orthonormal basis for Alice's subsystem.
	Taking trace then gives
	\begin{align*}
		\mathrm{tr}((I-A_{d+1})^{-1}A_i)=1.
	\end{align*}
	Notice that $(I-A_{d+1})^{-1}=I+(d+1)A_{d+1}$, {which can be seen from $\mathrm{tr}(A_{d+})=\frac{d}{d+1}$ and $\mathrm{rank}(A_{d+1})=1$, and thus} $\mathrm{tr}(A_{d+1}A_i)=\frac{1}{(d+1)^2}$.
	We now attempt to obtain the amplitudes of $A_{d+1}$ on the basis $\{e_i\}$:
	\begin{align*}
		\mathrm{tr}(A_{d+1}e_i e_i^{\dagger})= & \mathrm{tr}(A_{d+1}(I-A_{d+1})^{-1/2}A_i(I-A_{d+1})^{-1/2}) \\
		= & \mathrm{tr}((I+(d+1)A_{d+1})A_{d+1}A_i) \\
		= & \mathrm{tr}((d+1)A_{d+1}A_{i}) \\
		= & \frac{1}{d+1},
	\end{align*}
	where the second equality is due to $A_{d+1}$ commuting with $(I-A_{d+1})^{-1/2}$. Therefore, $A_{d+1}=\frac{1}{d+1}vv^{\dagger}$
	where $v=\sum_{i=1}^{d}\exp(\mathbf{i}\theta_i) e_i$ and $\theta_i$ are some real numbers.
	Consider the unitary $U=\sum_{i=1}^{d}\exp(-\mathbf{i}\theta_i) e_i e_i^{\dagger}$.
	Then
	\begin{align*}
		UA_{d+1}U^{\dagger}= & \frac{1}{d+1}w w^{\dagger}, \\
		UA_{i}U^{\dagger}= & U(I-A_{d+1})^{1/2}e_i e_i^{\dagger} (I-A_{d+1})^{1/2} U^{\dagger} \\
		= & \left[U(I-A_{d+1})^{1/2}U^{\dagger}\right](Ue_i e_i^{\dagger}U^{\dagger})\left[U(I-A_{d+1})^{1/2}U^{\dagger}\right] \\
		= & {\left[U\left(I+\frac{\sqrt{d+1}-(d+1)}{d}A_{d+1}\right)U^{\dagger}\right]e_i e_i^{\dagger}\left[U\left(I+\frac{\sqrt{d+1}-(d+1)}{d}A_{d+1}\right)U^{\dagger}\right]} \\
		= & \left(e_i-\frac{d+1-\sqrt{d+1}}{d(d+1)}w\right)\left(e_i-\frac{d+1-\sqrt{d+1}}{d(d+1)}w\right)^{\dagger} \\
		= & \left(e_i-\gamma_d w\right)\left(e_i-\gamma_d w\right)^{\dagger},
	\end{align*}
	where $i\in[d]$, $w=\sum_{i=1}^{d}e_i$ and $\gamma_d\equiv \frac{d+1-\sqrt{d+1}}{d(d+1)}$.

	We now perform the process prescribed above to Bob's system so that the measurement operators $B_i$ can be specified as well.
	After all the local unitaries are performed to the measurement operators and the state $\rho$, {the measurement operators can be written as}:
	\begin{align}
		\label{correlation:dp1measspec1}
		A_i= & \left(\ket{i}_A-\gamma_d\ket{\psi}_A\right)\left(\bra{i}_A-\gamma_d \bra{\psi}_A\right) & \forall i\in[d], \\
		\label{correlation:dp1measspec2}
		A_{d+1}= & \frac{1}{d+1}\ket{\psi}_A \bra{\psi}_A, & \\
		\label{correlation:dp1measspec3}
		B_i= & \left(\ket{i}_B-\gamma_d\ket{\psi}_B\right)\left(\bra{i}_B-\gamma_d \bra{\psi}_B\right) & \forall i\in[d], \\
		\label{correlation:dp1measspec4}
		B_{d+1}= & \frac{1}{d+1}\ket{\psi}_B \bra{\psi}_B, &
	\end{align}
	where
	\begin{align*}
		\ket{\psi}_A= & \sum_{i=1}^{d}\ket{i}_A, \\
		\ket{\psi}_B= & \sum_{i=1}^{d}\ket{i}_B,
	\end{align*}
        {and we record $e_i$ as $\ket{i}$ from now on}.

	\begin{lemm}
		\label{lemm:singularity}
		$(\bra{i}_A\otimes I_B)\rho(\ket{i}_A\otimes I_B)=\frac{1}{d}\ket{i}_B\bra{i}_B$.
	\end{lemm}
	\begin{proof}
		Let $X={d}(\bra{i}_A\otimes I_B)\rho(\ket{i}_A\otimes I_B)$;
		we shall show that $X=\ket{i}_B\bra{i}_B$.
		By Lemma \ref{lemm:mirror} we have
		\begin{align*}
			\mathrm{tr}(B_j X)= & d\mathrm{tr}((\ket{i}_A\bra{i}_A\otimes B_j)\rho) \\
			= & d\mathrm{tr}(\ket{i}_A\bra{i}_A\mathrm{tr}_B((I_A\otimes B_j)\rho)) \\
			= & \mathrm{tr}(\ket{i}_A\bra{i}_A A_j) \\
			= & \mathrm{tr}(B_j \ket{i}_B\bra{i}_B)
		\end{align*}
		for all $j\in[d+1]$.
		Consider the following primal optimization problem:
		\begin{equation}\nonumber
		    \begin{aligned}
			\text{minimize }\quad & \mathrm{tr}(X \ket{i}_B\bra{i}_B) \\
			\text{subject to }\quad & \forall j\in[d+1]\Big(\mathrm{tr}(B_j X)=\mathrm{tr}(B_j \ket{i}_B\bra{i}_B)\Big) \\
			& X\geq 0
		\end{aligned}.
		\end{equation}
		Its dual problem is
		\begin{equation}\nonumber
		    \begin{aligned}
			\text{maximize }\quad & -\sum_j y_j\mathrm{tr}(\ket{i}_B\bra{i}_B B_j) \\
			\text{subject to }\quad & \ket{i}_B\bra{i}_B+\sum_j y_j B_j\geq 0
		\end{aligned}.
		\end{equation}
		Let
		\begin{align}
			\xi_d= & \frac{(d-1)\gamma_d-1}{(1-\gamma_d)(1-d\gamma_d)}, \\
			\zeta_d= & \frac{1}{1-d\gamma_d}.
		\end{align}
		{We now show that}
		\begin{align*}
			y_j= &
			\begin{cases}
				\xi_d, & j=i \\
				\zeta_d, & j\neq i {\text{ and }} j\leq d \\
				0, & j=d+1 \\
			\end{cases}
		\end{align*}
		is a dual feasible solution with {the objective function of value $1$. By the} construction we have
		\begin{align*}
			\ket{i}_B\bra{i}_B+\sum_j y_j B_j=\sum_{jk}b_{jk}\ket{j}_B\bra{{k}}_B,
		\end{align*}
		where
		\begin{align*}
			b_{jk}= &
			\begin{cases}
				0, & { j \text{ or k} =i} \\
				(\xi_d+\zeta_d(d-2))\gamma_d^2+\zeta_d(1-\gamma_d)^2, & {j=k\neq i} \\
				{(\xi_d+\zeta_d{(d-3)})\gamma_d^2-2\zeta_d\gamma_d(1-\gamma_d),} & {j\neq k \text{ and }j,k \neq i}
			\end{cases}.
		\end{align*}
		First, we have
		\begin{align*}
			&{(\xi_d+\zeta_d(d-3))\gamma_d^2-2\zeta_d\gamma_d(1-\gamma_d)}\\
			= & \left(\frac{(d-1)\gamma_d-1}{(1-\gamma_d)(1-d\gamma_d)}+\frac{d-1}{1-d\gamma_d}\right)\gamma_d^2- \frac{2\gamma_d}{1-d\gamma_d}\\
			= & \frac{\gamma_d}{(1-\gamma_d)(1-d\gamma_d)}\left((d-1)\gamma_d^2-\gamma_d+{d\gamma_d-d\gamma_d^2-\gamma_d+\gamma_d^2}-2+2\gamma_d\right)\\
                = & {\frac{\gamma_d}{(1-\gamma_d)(1-d\gamma_d)}\left(d\gamma_d-2\right)\leq 0,}
		\end{align*}
		where in the last line we employ the relation $\gamma_d\leq \frac{1}{d}$;
		this shows that the off-diagonal entries of the matrix defined by $\{b_{jk}\}_{j\neq k}$ are nonpositive.
		Then, by the following relation
		\begin{align*}
			& (\xi_d+\zeta_d(d-2))\gamma_d^2+\zeta_d(1-\gamma_d)^2+(d-2)\big[(\xi_d+\zeta_d(d-3))\gamma_d^2-2\zeta_d\gamma_d(1-\gamma_d)\big] \\
			= & \frac{1-(d-1)\gamma_d}{1-\gamma_d}\\
			\geq & 0,
		\end{align*}
		we conclude that the matrix representation of $\ket{i}_B\bra{i}_B+\sum_j y_j B_j$ is diagonally dominant,
		hence $\ket{i}_B\bra{i}_B+\sum_j y_j B_j$ is positive semidefinite, and we obtain the desired dual feasibility.
        Meanwhile, it can be verified that $-\sum_j y_j\mathrm{tr}(\ket{i}_B\bra{i}_B B_j)=1$ for any $i\in[d]$. According to the weak duality property, for any feasible $X$ in the primal optimization problem, it holds that $\mathrm{tr}(X \ket{i}_B\bra{i}_B)\geq1$. Since
		\begin{align}
			\mathrm{tr}(X)=\sum_{j=1}^{d+1}\mathrm{tr}(B_j X)=\sum_{j=1}^{d+1}\mathrm{tr}(B_j \ket{i}_B\bra{i}_B)=\mathrm{tr}(\ket{i}_B\bra{i}_B)=1,
		\end{align}
		this implies that $X=\ket{i}_B\bra{i}_B$.
	\end{proof}

	\begin{prop}
		$\rho={\frac{1}{d}}\sum_{ij}\ket{i}_A\ket{i}_B\bra{j}_A\bra{j}_B$.
	\end{prop}
	\begin{proof}
		By Lemma \ref{lemm:singularity},
		the state $\rho$ lies in the support of $\{\ket{1}_A\ket{1}_B,\cdots,\ket{d}_A\ket{d}_B\}$,
		which means that $\rho$ can be written as
		\begin{align*}
			\rho= & \sum_{jk}(\bra{j}_A\bra{j}_B\rho\ket{k}_A\ket{k}_B)\ket{j}_A\ket{j}_B\bra{k}_A\bra{k}_B.
		\end{align*}
		We now evaluate $\tilde{B}_i$ with this expansion of $\rho$ in order to apply Lemma \ref{lemm:mirror}:
		\begin{align*}
			\tilde{B}_i= & \mathrm{tr}_B((I_A\otimes B_i)\rho) \\
			= & \mathrm{tr}_B\left((I_A\otimes B_i)\sum_{jk}(\bra{j}_A\bra{j}_B\rho\ket{k}_A\ket{k}_B)\ket{j}_A\ket{j}_B\bra{k}_A\bra{k}_B\right) \\
			= & \sum_{jk}(\bra{j}_A\bra{j}_B\rho\ket{k}_A\ket{k}_B)(\bra{k}_B B_i\ket{j}_B)\ket{j}_A\bra{k}_A \\
			= & \sum_{jk}(\bra{j}_A\bra{j}_B\rho\ket{k}_A\ket{k}_B)(\bra{k}_A A_i\ket{j}_A)\ket{j}_A\bra{k}_A.
		\end{align*}
		Lemma \ref{lemm:mirror} states that
		\begin{align*}
			\tilde{B}_i= & {\frac{1}{d}} A_i,
		\end{align*}
		hence
		\begin{align*}
			(\bra{j}_A\bra{j}_B\rho\ket{k}_A\ket{k}_B)(\bra{k}_A A_i\ket{j}_A)=\bra{j}_A\tilde{B}_i\ket{k}_A={\frac{1}{d}} \bra{j}_A A_i\ket{k}_A.
		\end{align*}
		By the construction of the measurement operators $A_i$ we have
		\begin{align*}
			\bra{k}_A A_i\ket{j}_A=\bra{j}_A A_i\ket{k}_A\neq 0,
		\end{align*}
		hence the above identity implies
		\begin{align}
			\bra{j}_A\bra{j}_B\rho\ket{k}_A\ket{k}_B={\frac{1}{d}},
		\end{align}
		which means that $\rho$ is a $d\times d$ maximally entangled pure state.
	\end{proof}
    {This way, we have proved that if Alice and Bob generate the correlation $R$ in Eq.\eqref{correlation:maximal2} by locally measuring a $d\times d$ bipartite quantum state $\rho$, then $\rho$ must be a maximally entangled pure state, where the two local quantum measurements in Eqs.\eqref{correlation:dp1measspec1}-\eqref{correlation:dp1measspec4} can fulfill this task.}

    We now utilize the rescaling technique to certify general bipartite pure states with $(d+1)\times(d+1)$ correlations.
	\begin{theo}
		\label{correlation:generaldp1}
		Pick real numbers $\alpha_i>0$ and $\beta_i>0$ for $i=1,\cdots,d+1$.
		Let
		\begin{align*}
			\vec{\alpha}= & (\alpha_1,\cdots,\alpha_{d+1}), \\
			\vec{\beta}= & (\beta_1,\cdots,\beta_{d+1}).
		\end{align*}
		If
		\begin{align*}
			\tilde{R}= & \mathrm{diag}(\vec{\alpha})R\mathrm{diag}(\vec{\beta})
		\end{align*}
		is a valid correlation and can be obtained by locally measuring a $d\times d$ state $\rho$, then $\rho$ is a bipartite pure entangled state with the squared Schmidt coefficients same as the spectrum of the matrix
		\begin{align}
			\left(\sum_{i=1}^{d+1} \frac{\alpha_i}{\sqrt{d}} A_i\right)\left(\sum_{j=1}^{d+1} \frac{\beta_j}{\sqrt{d}} A_j\right),
		\end{align}
		where $A_i$ are defined as in {Eqs.}(\ref{correlation:dp1measspec1}) and (\ref{correlation:dp1measspec2}).
	\end{theo}

	\begin{proof}
		Notice that $\tilde{R}$ has the same PSD rank $d$ as $R$.
		Let $\{C_i\}$, $\{D_j\}$ be a $d\times d$ PSD factorization for $\tilde{R}$.
		Then $\{\alpha_i^{-1}C_i\}$ and $\{\beta_j^{-1}D_j\}$ form a $d\times d$ PSD factorization for $R$,
		{this} means that the PSD factorization $\{\alpha_i^{-1}C_i\}$, $\{\beta_j^{-1}D_j\}$
		is equivalent to a canonical form of PSD factorization $\{C'_i\}$, $\{D'_j\}$ of {$R$}
		such that $C'_i=D'_i=A_i/\sqrt{d}$ for any $i\in[d+1]$.
		The same equivalence transformation sends $\{C_i\}$, $\{D_j\}$
		to $\{\alpha_i C'_i\}$, $\{\beta_j D'_j\}$,
		hence the spectrum of the canonical form of PSD factorization for $\tilde{R}$ that is equivalent to $\{C_i\}$ and $\{D_j\}$ is the same as that of
		\begin{align*}
			S=\left(\sum_{i=1}^{d+1} \frac{\alpha_i}{\sqrt{d}} A_i\right)\left(\sum_{j=1}^{d+1} \frac{\beta_j}{\sqrt{d}} A_j\right),
		\end{align*}
		where we have utilized Lemma \ref{correlation:scsd}. Therefore, any bipartite $d\times d$ state $\rho$ that generates {$\tilde{R}$}
		must have its reduced states $\mathrm{tr}_A(\rho)$ and $\mathrm{tr}_B(\rho)$
		having the spectrum equal to that of $S$.

		Now let $A_i$ and $B_j$ be the local measurement operators that generate {$\tilde{R}$} upon measuring $\rho$;
		that is,
		\begin{align*}
			\mathrm{tr}((A_i\otimes B_j)\rho)={\tilde{R}_{ij}}.
		\end{align*}
		Then
		\begin{align*}
			\mathrm{tr}((\vec{\alpha}_i^{-1} A_i\otimes \vec{\beta}_j^{-1} B_j)\rho)={R_{ij}}.
		\end{align*}
		Let $S_A=\sum_{i=1}^{2d}\vec{\alpha}_i^{-1}A_i$ and $S_B=\sum_{j=1}^{2d}\vec{\beta}_j^{-1}B_j$.
		Then $\{S_A^{-1/2}(\vec{\alpha}_i^{-1} A_i)S_A^{-1/2}\}$ and $\{S_B^{-1/2}\vec{\beta}_j^{-1} B_j S_B^{-1/2}\}$
		are the operators of two valid quantum measurements as well;
		furthermore, they generate {$R$} upon measuring $(S_A^{1/2}\otimes S_B^{1/2})\rho(S_A^{1/2}\otimes S_B^{1/2})$, which can be verified to be a valid quantum state.
		This implies that the state $(S_A^{1/2}\otimes S_B^{1/2})\rho(S_A^{1/2}\otimes S_B^{1/2})$ is pure,
		hence $\rho$ is also pure.
		By the spectrum of the reduced density matrices of $\rho$, we conclude that $\rho$ is a bipartite pure state with the squared Schmidt coefficients same as the spectrum of $S$.
	\end{proof}

    Next we prove that general bipartite pure states can be certified with $(d+1)\times(d+1)$ correlations in an approximate manner, and we conjecture that exact certification for these quantum states is also possible.
	\begin{prop}
		The set of bipartite pure quantum states that can be certified via
		the construction in Theorem \ref{correlation:generaldp1} has dense Schmidt coefficient vectors $(\sqrt{\lambda_1},\cdots,\sqrt{\lambda_d})^T$ in the set
		\begin{align*}
			K=\left\{\vec{\lambda}\in\mathbb{R}^d|\lambda_i\geq 0,\sum_{i=1}^{d}\lambda_i=1,\lambda_1\geq\cdots\geq\lambda_d\right\}.
		\end{align*}
	\end{prop}

	\begin{proof}
		Suppose $(\lambda_1,\cdots,\lambda_d)\in K$.
		For any sufficiently small $\varepsilon>0$,
		let
		\begin{align*}
			\alpha_i= & \left(1-\frac{\varepsilon}{(d+1)^2}\right)(d\lambda_i), & \forall i\in[d], \\
			\alpha_{d+1}= & \frac{\varepsilon}{(d+1)^2}, \\
			\beta_i= & 1, & \forall i\in[d], \\
			\beta_{d+1}= & {d+2}.
		\end{align*}
		Then it can be verified that the resulting $\tilde{R}$ in Theorem \ref{correlation:generaldp1} is a valid correlation. The spectrum it certifies is then the same as that of the operator
		\begin{align*}
			S\equiv & \frac{1}{d}\left(\sum_{j=1}^{d+1} \beta_j A_j\right)^{{1/2}}\left(\sum_{i=1}^{d+1} \alpha_i A_i\right)\left(\sum_{j=1}^{d+1} \beta_j A_j\right)^{{1/2}} \\
			{= }& {\frac{1}{d}\left(I + (d+1) A_{d+1}\right)^{1/2}\left(\sum_{i=1}^{d+1} \alpha_i A_i\right)\left(I + (d+1) A_{d+1}\right)^{1/2}.}
		\end{align*}
		{Recall that $A_{d+1}= \frac{1}{d+1}\ket{\psi}_A \bra{\psi}_A$ as given in Eq .(\ref{correlation:dp1measspec3}), then we have $\left(I + (d+1) A_{d+1}\right)$ is a matrix with eigenvalues $d+1,1,... ,1.$ and the corresponding eigenvectors $\frac{1}{\sqrt{d}}\ket{\psi}_A, \ket{\psi_1},...,\ket{\psi_{d-1}}$, where $\{\frac{1}{\sqrt{d}}\ket{\psi}_A, \ket{\psi_{1}},...,\ket{\psi_{d-1}}\}$ forms an orthonormal basis for $\mathbb{C}^d$. Consequently, $\left(I + (d+1) A_{d+1}\right)^{1/2}$ can be expressed as $I + \frac{(d+1)(\sqrt{d+1}-1)}{d} A_{d+1}$. By straightforward calculation, it can be shown that for any $i\in [d]$, we have
		$$\left(I + (d+1) A_{d+1}\right)^{1/2}A_i \left(I + (d+1) A_{d+1}\right)^{1/2} = \ket{i}\bra{i}, $$
		hence $S$ could be written as
		\begin{align}
			S=&\left(1-\frac{\varepsilon}{(d+1)^2}\right)\mathrm{diag}(\lambda_1,\cdots,\lambda_d)+\frac{\varepsilon}{d(d+1)}A_{d+1}\\
            = &  \mathrm{diag}(\lambda_1,\cdots,\lambda_d) + H,
		\end{align}
		where $H$} is Hermitian with its operator norm $\lVert H\rVert$ satisfying
		$\lVert H\rVert\leq \frac{2}{(d+1)^2}\varepsilon$.
		By Weyl's perturbation theorem, we have
		\begin{align}
			\max_j |\mu^{\downarrow}_j(S)-\lambda_j|\leq \lVert H\rVert\leq \frac{2}{(d+1)^2}\varepsilon,
		\end{align}
		where $\mu^{\downarrow}_j(\cdot)$ denotes the $j$-th largest eigenvalue of the enclosed Hermitian operator.
		This shows that the spectrum of $S$ is within $\frac{2}{(d+1)^2}\varepsilon$ distance to $(\lambda_1,\cdots,\lambda_d)$ in sense of $\ell_{\text{op}}$ norm,
		which completes the proof.
	\end{proof}

    \bibliography{ref}